\documentclass[journal]{IEEEtran}
\usepackage{cite}

\ifCLASSINFOpdf
  \usepackage[pdftex]{graphicx}
  \DeclareGraphicsExtensions{.pdf,.jpeg,.png}
\else
\fi
\usepackage{amsmath}
\usepackage{subcaption}

\usepackage{amssymb}  

\usepackage{amsthm}
\theoremstyle{plain}
\newtheorem{thm}{Theorem}
\newtheorem{prop}{Proposition}
\newtheorem{lem}{Lemma}

\theoremstyle{definition}
\newtheorem{defn}{Definition}

\theoremstyle{remark}
\newtheorem{rem}{Remark}

\newcommand{\Tr}{\textup{Tr}}     

\usepackage{ifthen}
\newboolean{TwoColEq}
\setboolean{TwoColEq}{true} 
\ifthenelse {\boolean{TwoColEq}} 
{
\usepackage{flushend}
}{
}

\usepackage{pgf}
\usepackage[normalem]{ulem}
\usepackage[textsize=footnotesize,textwidth=3cm]{todonotes} 

\begin{document}
%
\title{A Two-Stage Architecture for Differentially Private Kalman Filtering and LQG Control} 

%
%

\author{Kwassi~H.~Degue
        and~Jerome~Le Ny,~\IEEEmembership{Senior Member,~IEEE}
\thanks{A preliminary version of this paper appeared in \cite{DegueGlobalSIP}. 
This work was supported by NSERC under Grant RGPIN-5287-2018 and RGPAS-2018-522686, 
by the Pierre Arbour Foundation doctoral scholarship and by an FRQNT doctoral scholarship.
}
\thanks{The authors are with the Department of Electrical Engineering,
Polytechnique Montreal and GERAD, Montreal, QC H3T-1J4, Canada 
{\tt\small \{kwassi-holali.degue, jerome.le-ny\}@polymtl.ca}.}}

%
%

\markboth{}%
{Degue et al.: Differentially Private Kalman Filtering and LQG Control}
%



\maketitle 

\begin{abstract}
Large-scale monitoring and control systems enabling a more intelligent infrastructure 
increasingly rely on sensitive data obtained from private agents, e.g., location 
traces collected from the users of an intelligent transportation system. 
In order to encourage the participation of these agents, it becomes then critical 
to design algorithms that process information in a privacy-preserving way. 
This article revisits the Kalman filtering and Linear Quadratic Gaussian (LQG) 
control problems, subject to privacy constraints. 
We aim to enforce differential privacy, a formal, state-of-the-art 
definition of privacy ensuring that the output of an algorithm is not too 
sensitive to the data collected from any single participating agent.
A two-stage architecture is proposed that first aggregates and combines the individual 
agent signals before adding privacy-preserving noise and post-filtering the result to 
be published. We show a significant performance improvement offered by this architecture 
over input perturbation schemes as the number of input signals increases and that 
an optimal static aggregation stage can be computed by solving a semidefinite program.
The two-stage architecture, which we develop first for Kalman filtering, is then 
adapted to the LQG control problem by leveraging the separation principle.
Numerical simulations illustrate the performance improvements over differentially 
private algorithms without first-stage signal aggregation. 
\end{abstract}


%

\section{Introduction}

To monitor and control intelligent infrastructure systems such as smart grids,
smart buildings or smart cities, data needs to be continuously collected 
from the people interacting with these systems, either through sensors installed 
in the environment such as cameras and smart meters, or through personal devices 
such as smartphones. 
Hence, in contrast to more traditional control systems, the measured signals for
such systems often contain highly privacy-sensitive information, e.g., related 
to the real-time location or health of a person.
For example, the accuracy of crowd-sourced traffic maps and congestion-aware routing 
applications is increased by using data provided by smartphones and connected vehicles
\cite{Herrera:TR09:mobileCentury}.
However, individual location data turns out to be very difficult to properly 
anonymize because individuals have highly unique mobility patterns 
\cite{Shokri2011,Montjoye:SR13:locationPrivacy}, 
and in fact individual trajectories can be reconstructed even from 
just aggregate location data \cite{Xu:WWW17:trajectoryReconstruction,Pyrgelis:PET17:privacyEvaluation}.
Similarly, fine-grained measurements of a house's electric power consumption 
collected by a smart meter can enable demand-response schemes, but can 
also be used to infer the activities of the occupants, 
by identifying the usage of individual appliances 
\cite{Hart:IEEE92:loadMonitoring,Bauer2009,Lisovich2010,Molina:workshop10:smartMeters}.
Therefore, it is necessary to implement privacy-preserving mechanisms 
when sensitive data must be shared to improve a system's performance.

Various definitions of privacy have been proposed that are amenable to formal analysis.
While a survey of such definitions is out of the scope of this paper, we can mention 
some recent work focusing on signal processing and control problems.
Privacy is measured by a \emph{lower bound} on the mutual information between published
and private signals in \cite{Sankar2011}, 
on the Fisher information in \cite{Farokhi:book20:fisherPrivacy}, or on the error covariance 
of the estimator of a sensitive signal 
in \cite{Mo:TAC16:privateConsensus, Song:TIFS20:compressivePrivacy}.
The concept of $k$-anonymity and its extensions has been applied to the publication 
of location traces in \cite{Fei2019}.
But much of the recent research on privacy-preserving data analysis relies on the notion
of \textit{differential privacy} \cite{Dwork2006,Dwork2006_new,DworkBook}. 
In the standard set-up, which is also the situation considered in this article, a data
holder aims to release the results of computations based on private data. Differential
privacy is enforced by adding an appropriate amount of noise to the published results, 
in such a way that the probability distribution over the outputs does not depend too 
much on the data of any single individual.
As a result, the ability of a third party observing the outputs to make 
new inferences about a given person is roughly the same, whether or not that
person chooses to contribute its data. This guarantee can then be used to 
weigh the risks of information disclosure against the benefits of publishing 
more accurate analyses.

A large number of techniques have been developed to compute differentially private
versions of various statistics, see \cite{DworkBook} for an overview.
Nevertheless, the differentially private analysis of \emph{streaming} data remains
relatively less explored \cite{Dwork2010,Sankar2011,Fan2014}, despite its 
importance for signal processing and control applications. 
Some previous work has focused on the design of differentially private dynamic estimators \cite{LeNyTAC2014,LeNy:TAC18:MIMOdpef,LeNy:IJRNC18:dpContraction,Mcglinchey:ECC18:dpPositiveObs}, 
controllers \cite{WangIEE2017,HaleACC2018}, 
consensus algorithms \cite{Huang:workshop12:dpConsensus,Nozari:Automatica17:dpConsensus},
or anomaly detectors 
\cite{Cummings:NIPS18:dpDetection,DegueAllerton2018,Rostampour:SAFEPROCESS18:dpDetection}.
In particular, \cite{LeNyTAC2014} discusses the Kalman filtering problem under 
a differential privacy constraint and compares schemes introducing noise either 
directly on the measured signals (input perturbation mechanisms) or on the 
published estimate (output perturbation mechanisms). Output perturbation provides
better performance as the number of input signals increases, but has the
drawback of leaving unfiltered noise on the output, which motivates the 
two-stage architecture that we consider here.
In \cite{WangIEE2017}, the authors consider a multi-agent linear quadratic tracking 
problem where the trajectory tracked by each agent should remain private, while
\cite{HaleACC2018} considers an LQG control problem where each agent wishes 
to keep its individual state private. In both cases, noise is added directly
on the individual measurements, a form of input perturbation.

In this paper we study the design of Kalman filters and LQG controllers subject 
to a differential privacy constraint on the measured signals. 
These problems, stated formally in Section \ref{sec:statement}, arise when a data
collector measures private signals originating from a population of agents, whose 
dynamics can be modeled as linear Gaussian systems, in order to publish in real-time 
either an estimate of an aggregate state of the agent population, or a control signal 
shared by the agents and aimed at regulating such an aggregate state.
As a motivating example, one can consider the problem of controlling the distribution
of vehicles on a road network by means of traffic messages broadcasted to all cars, with
the current density estimated from location data obtained from the smartphones of 
individual drivers.

Section \ref{sec:DP Kalman Filtering} presents the main contribution of this paper,
namely, a two-stage architecture for differentially private Kalman filtering, 
where the privacy-preserving noise is added only after an input stage appropriately 
combining the measured signals of the individual agents, while an output stage filters 
out this noise.
Such two-stage architectures were discussed in \cite{LeNyTAC2014} but have not yet
been applied to the Kalman filtering problem, and we argue first in 
Section \ref{sec:Main-results-Scalar} via a simple example that significant performance 
improvements can be expected compared to input perturbation mechanisms. 
We show that the optimal input stage can be computed by solving a semidefinite
program (SDP), hence, a tractable convex optimization problem. 
The fact that the input stage design problem admits an SDP formulation 
is reminiscent of other Kalman filtering problems subject to resource 
or communication constraints, see, e.g., 
\cite{LeNy:TAC11:KFscheduling,Mourikis:TRO06:optimalSensorScheduling,Tanaka2017},
but the SDP capturing specifically the differential privacy constraint is new.
The design procedure is then adapted in Section \ref{sec:ProblemFormulationControl} 
to the LQG control problem. By exploiting the classical properties of the optimal
LQG controller (linearity and separation principle), we can view the control problem 
as the problem of estimating a certain linear combination of the agent states as 
in Section \ref{sec:DP Kalman Filtering}, but for a specific cost on the estimation error.

A conference version of this paper appeared in \cite{DegueGlobalSIP}, but contained no proof
and did not discuss the LQG control problem. 
The numerical examples discussed in Sections \ref{sec:Examples} and \ref{sec:ControlNumericalexperiments} 
are also new.
Finally, we introduce some notation used throughout this paper. 
We fix a generic probability triple $(\Omega,\mathcal{F},\mathbb{P})$, where $\mathcal{F}$ 
is a $\sigma$-algebra on $\Omega$ and $\mathbb{P}$ a probability measure
defined on $\mathcal F$. 
The notation $X \sim \mathcal N(\mu,\Sigma)$ means that $X$ is a Gaussian 
random vector with mean vector $\mu$ and covariance matrix $\Sigma$.
``Independent and identically distributed'' is abbreviated iid.
We denote the $p$-norm of a vector $x\in\mathbb{R}{}^{k}$ by 
$|x| _{p}:=(\sum_{i=1}^{k}|x_{i}|^{p})^{1/p}$, for $p \in [1,\infty)$. 
For a matrix $A$, the induced 2-norm (maximum singular value of $A$) is denoted 
$\|A\|_2$ and the Frobenius norm $\|A\|_F := \sqrt{\text{Tr}(A^{\text{T}}A)}$.
If $A$ and $B$ are symmetric matrices, $A \succeq B$ (resp. $A \succ B$) means that 
$A-B$ is positive semi-definite (resp. positive definite).
We use the notation $\text{diag}(A_1,\ldots,A_n)$ to represent a block-diagonal
matrix with the matrices $A_i$ on the diagonal. The column vector of size $n$ with all
components equal to $1$ is denoted $\mathbf 1_n$.
Finally, for a discrete-time signal $x$ we denote $x_{0:t} := \{x_0,\ldots,x_t\}$.

\section{Problem Statement} 
\label{sec:statement} 

\subsection{Privacy-Preserving State Estimation and LQG Control for a Population of Dynamic Agents}
\label{section: estimation and control problems}

Consider a set of $n$ privacy-sensitive signals $\{y_{i,t}\}_{0 \leq t \leq T}$, 
$i=1,\ldots,n$, with $y_{i,t} \in \mathbb R^{p_i}$, collected by a data aggregator, 
and which could originate from $n$ distinct agents. Let $p = \sum_{i=1}^n p_i$.
We assume that a mathematical model capturing known dynamic and statistical properties
of these signals is publicly available, consisting of a linear system with $n$ independent 
(vector-valued) states associated to the $n$ measured signals
\begin{align}
\begin{split}  \label{eq:IndividualGeneral}
x_{i,t+1} &= A_{i,t} \, x_{i,t} +B_{i,t} \, u_{t}+ w_{i,t}, \;\; 0 \leq t \leq T-1, \\
y_{i,t} &= C_{i,t} \, x_{i,t} + v_{i,t}, \;\; 0 \leq t \leq T,    
\end{split}
\end{align}
for $i=1,\ldots,n$, where $x_{i,t}, w_{i,t} \in \mathbb R^{m_i}$.
Here $w_{i,t} \sim \mathcal{N}(0,W_{i,t})$ and $v_{i,t} \sim \mathcal{N}(0,V_{i})$ are 
independent sequences of iid zero-mean Gaussian random vectors with covariance matrices 
$W_{i,t} \succ 0, V_{i} \succ 0$, for $i=1, \ldots, n$. 
In particular, assuming that the matrices $W_{i,t}$ are invertible is necessary in the 
following to be able to use the ``information filter'' form of the Kalman filter 
equations \cite{Anderson2005}.
The sequence $u$ with $u_{t}\in\mathbb{R}^{h}$ represents a control input that is shared 
by the $n$ individuals. This is motivated by scenarios in which a common signal is broadcast
to drive the aggregate state of a population, while individual signals can still be subject to
privacy constraints. 
The initial conditions $x_{i,0}$ are independent Gaussian random vectors that are also 
independent of the noise processes $w$ and $v$, with mean $\overline{x}_{i,0}$ 
and covariance matrices $\Sigma_{i,0}^{-} \succ 0$. 
Let $x_t := [x_{1,t}^T,\ldots,x_{n,t}^T]^T$, $y_t := [y_{1,t}^T,\ldots,y_{n,t}^T]^T$, 
$w_t = [w_{1,t}^T,\ldots,w_{n,t}^T]^T$ and  $v_t = [v_{1,t}^T,\ldots,v_{n,t}^T]^T$ denote 
the global state, measurement and noise signals of \eqref{eq:IndividualGeneral}. 
Define $A_t := \textup{diag}(A_{1,t},\ldots,A_{n,t})$, $B_t = [B_{1,t}^T,\ldots,B_{n,t}^T]^T$, $C_t = \textup{diag}(C_{1,t},\ldots,C_{n,t})$,
$W_t = \textup{diag}(W_{1,t},\ldots,W_{n,t})$, $V = \textup{diag}(V_1,\ldots,V_n)$. Then the system \eqref{eq:IndividualGeneral}
can be rewritten more compactly as
\begin{align}
x_{t+1} & = A_t \, x_{t}+ B_t \, u_t + w_{t},	 \;\; 0 \leq t \leq T-1,	\label{eq:Global}  \\
y_{t} & = C_t \, x_{t}+v_{t}, \;\; 0 \leq t \leq T,	\label{eq:Global measurements}
\end{align}
with $w_t \sim \mathcal N(0,W_t)$ and $v_t \sim \mathcal N(0,V)$. Throughout the paper,
the model parameters $\overline{x}_{i,0}$, $\Sigma^-_{i,0}$, $A_{i,t}$, $B_{i,t}$, $C_{i,t}$, 
$W_{i,t}$, $V_i$, are assumed to be publicly known information.

In Section \ref{sec:DP Kalman Filtering}, we first consider a filtering problem (with $u$ a known signal) 
where the data aggregator aims to publish at each period $t$ a causal estimate $\hat z_t$ of 
a linear combination $z_{t}= L_t x_t = \sum_{i=1}^{n}L_{i,t}x_{i,t}$ of the individual states, 
computed from the signals $y_{i}$, with $L_t := \begin{bmatrix} L_{i,t}, \ldots, L_{n,t} \end{bmatrix}$ 
some given (publicly known) matrices. 
This estimator should minimize the Mean Square Error (MSE) performance measure
\begin{equation}	\label{eq: MSE performance}
E_T := \frac{1}{T+1} \sum_{t=0}^T \mathbb E \left[\|z_t - \hat z_t\|_2^2\right].
\end{equation}
For privacy reasons, the signals $y_{i}$ are not released by the data aggregator and moreover 
the publicly released estimate $\hat z$ should also guarantee the differential privacy of the input 
signals $y_i$, as defined precisely in Section \ref{section: Differential privacy background}.
For example, the signals $y_i$ could represent position measurements of $n$ individuals, 
each state $x_i$ could consist of the position and velocity of individual $i$, 
and the goal might be to publish only a real-time estimate of the average velocity 
of all individuals.
Note that \emph{in the absence of privacy constraint}, 
the optimal estimator is $\hat{z}_{t}=\sum_{i=1}^{n}L_{i,t} \hat{x}_{i,t}$, with $\hat{x}_{i,t}$ 
provided by the (time-varying) Kalman filter estimating the state $x_i$ of subsystem $i$ from the 
signal $y_i$ \cite{Anderson2005}, and in particular the estimation problem then decouples for the 
$n$ subsystems.

Next, in Section \ref{sec:ProblemFormulationControl} we build on the results obtained for
the filtering problem to study the following privacy-constrained Linear Quadratic Gaussian (LQG) 
regulation problem. The data aggregator uses the measured
signals $y_i, 1 \leq i \leq n$, to compute and broadcast a (causal) control signal $u$ that minimizes 
the following quadratic cost for the $n$ agents
\begin{equation}	\label{eq: cost function LQG}
J_T= \frac{1}{T+1} \mathbb{\mathbb{E}} \left[ \sum_{t=0}^{T-1} 
\left(x_{t}^T Q_t x_{t}+u_{t}^T R_t u_{t}\right) + x_{T}^T Q_T x_{T} \right],
\end{equation}
where $Q_t \succeq 0$ for $0 \leq t \leq T$ and $R_t \succ 0$ for $0 \leq t \leq T-1$ are
publicly known weight matrices.
Again, only the signal $u$ is published (in particular, it is available to the $n$ agents), 
and releasing $u$ must guarantee the differential privacy of the measured signals $y_i$.
It is worth noting that the cost function \eqref{eq: cost function LQG} can be used to
drive an aggregate value of the global population state toward $0$ rather than the
individual agent states, since trying to do the latter might be in direct conflict with
the privacy requirement (which, essentially, aims to hide the value of the individual
signals $y_i$, and hence indirectly of the individual states $x_i$).
For example, we can have 
$x_t Q_t x_t = \left ( \frac{1}{n} \sum_{i=1}^n x_{i,t} \right)^2$ if $Q_t = \frac{1}{n^2} \mathbf 1_n \mathbf 1_n^T$,
in order to regulate the average population state.
Figure \ref{fig:problem setup} represents the estimation and control problem setups, including a basic
scheme to enforce differential privacy by injecting noise directly in the signals $\{y_i\}_{1 \leq i \leq n}$,
as described in Section \ref{section: input perturbation}.

\begin{figure}  
    \centering
    \begin{subfigure}[b]{0.23\textwidth}
        \includegraphics[width=\textwidth]{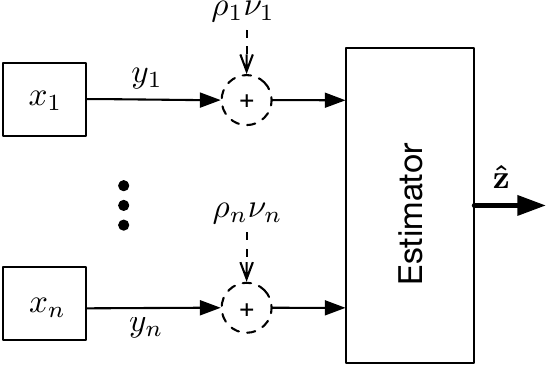}
        \caption{Estimation Problem}
        \label{fig:filtering setup}
    \end{subfigure}
    ~ 
    \begin{subfigure}[b]{0.23\textwidth}
        \includegraphics[width=\textwidth]{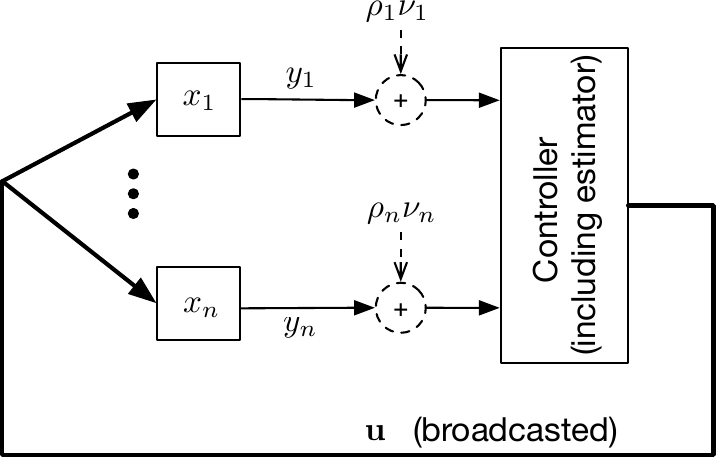}
        \caption{Control Problem}
        \label{fig:control setup}
    \end{subfigure}
    \caption{Estimation and control problem setup. Signals $\{y_i\}_{1 \leq i \leq n}$ produced by $n$ independent agents with states 
    $\{x_i\}_{1 \leq i \leq n}$ are collected for monitoring or control purposes. The thicker lines correspond to published signals (estimate 
    $\hat z$ or broadcasted control input $u$). 
    Dashed lines represent the injection of privacy-preserving noise by the input perturbation mechanism,
    a simple scheme to enforce differential privacy described in Section \ref{section: input perturbation}.}	\label{fig:problem setup}
\end{figure}

Note that the steady-state versions of both the filtering and LQG control problems are also considered, with the 
model \eqref{eq:Global}-\eqref{eq:Global measurements} in this case assumed time-invariant and the performance 
measures defined as
\begin{equation} \label{eq: asymptotic performance}
E_\infty := \lim_{T \to \infty} E_T, \;\; J_\infty := \lim_{T \to \infty} J_T.
\end{equation}
To finish stating the above problems formally, we define in the next section the differential privacy 
constraint imposed on the signal $\hat z$ for the filtering problem and $u$ for the LQG problem.

\subsection{Differential Privacy Constraint}
\label{section: Differential privacy background} 

Differential privacy \cite{DworkBook} is a property satisfied by certain randomized algorithms (also called 
mechanisms), which in an abstract setting compute outputs in a space $\mathsf R$ based on sensitive 
data in a space $\mathsf D$. To be differentially private, the algorithm must ensure that the probability
distribution of its randomized output is not very sensitive to certain variations in the input data, which
are specified as part of the privacy requirement.

More concretely, we equip the input space $\mathsf D$ with a symmetric binary relation called
adjacency and denoted $\text{Adj}$, which captures the variations in the input datasets that we want to
make hard to detect by observing the outputs.
In our case, the input space $\mathsf D$ is the vector space $\mathbb R^{p(T+1)}$ of global measurement 
signals $y$, and a mechanism is a causal stochastic system producing an output signal ($\hat z$ or $u$ 
in the previous section) based on its input $y$.
We define two measured signals to be adjacent if the following condition holds
\begin{align}	\label{eq:adjacency}
\begin{split}
\text{Adj}(y,y')\,\text{iff for some}\, 1 \leq i \leq n,\,\| y_{i}-y_{i}'\| _{2}\leq\rho_{i}, \\
\text{and}\,y_{j}=y_{j}'\,\text{for all}\,j\neq i, 
\end{split}
\end{align}
with $\{\rho_i\}_{i=1}^{n}\in\mathbb{R}_{+}^{n}$ a given set of positive numbers
and with the definition of the $\ell_2$-norm $\|v\|_2 := \left(\sum_{t=0}^T |v_t|_2^2\right)^{1/2}$ for
a vector-valued signal $v$.
In other words, two adjacent measurement signals can differ by the values of 
a single participant, with only $\ell_2$-bounded signal deviations allowed for each individual.
For any two inputs that satisfy the adjacency relation, the following definition characterizes 
the deviation that is allowed for a differentially private mechanism's output distribution. 

\begin{defn} \label{def:Differential privacy}
Let $\mathsf{D}$ be a space equipped with a symmetric binary relation denoted $\text{Adj}$ 
and let $(\mathsf{R},\ensuremath{\mathcal{R}})$ be a measurable space, where 
$\ensuremath{\mathcal{R}}$ is a given $\sigma$-algebra over $\mathsf{R}$. 
Let $\epsilon \geq 0$, $1 \geq \delta \geq 0$. A randomized mechanism $M$ from 
$\mathsf{D}$ to $\mathsf{R}$ is $(\epsilon,\delta)$-differentially private (for $\text{Adj}$) if 
for all $d,d'\in\mathsf{D}$ such that $\text{Adj}(d,d')$, 
\begin{equation}
\mathbb{P}(M(d)\in S) \leq e^\epsilon \, \mathbb{P}(M(d')\in S) + \delta, 
\, \forall S \in \mathcal{R}. \label{eq:privacydefinition}
\end{equation}
\end{defn}  

In Definition \ref{def:Differential privacy}, smaller values of $\epsilon$ and $\delta$ correspond 
to stronger privacy guarantees, i.e., distributions for $M(d)$ and $M(d')$ that are closer in
\eqref{eq:privacydefinition}. Next, we need tools that can be used to enforce the property of 
Definition \ref{def:Differential privacy}. The following definition is useful for mechanisms 
such as ours that produce outputs in vector spaces.

\begin{defn}	\label{def: sensitivity}
Let $\mathsf{D}$ be a space equipped with an adjacency relation $\text{Adj}$.  
Let $\mathsf R$ be a vector space equipped with a norm $\|\cdot\|_{\mathsf R}$. 
The sensitivity of a mapping $q:\mathsf{D} \mapsto \mathsf{R}$ is defined as
\[
\triangle q := \sup_{\{d,d':\text{Adj}(d,d')\}}
\| q(d)-q(d') \|_{\mathsf{R}}.
\]
For $\mathsf{R}=\mathbb{R}^{h(T+1)}$ or $(\mathbb R^h)^\mathbb N$ equipped with 
the $\ell_2$-norm, this defines the $\ell_{2}$-sensitivity of $q$, denoted $\triangle_{2}q$. 
\end{defn}

The \emph{Gaussian mechanism} \cite{Dwork2006_new1} consists in adding Gaussian noise 
proportional to the $\ell_2$-sensitivity of a mapping to enforce $(\epsilon,\delta)$-differential privacy. 
A fairly tight upper bound on the proportionality constant is provided in \cite{LeNyTAC2014}.
Recall first the definition of the $\mathcal{Q}$-function
$\mathcal{Q}(x):=\frac{1}{\sqrt{2\pi}}\int_{x}^{\infty}\exp(-\frac{u^{2}}{2})du$, which is monotonically 
decreasing from $(-\infty,\infty)$ to $(0,1)$.
Then, for $1 > \delta > 0$, define
$\kappa_{\delta,\epsilon} := \frac{1}{2\epsilon} \left( Q^{-1}(\delta)+\sqrt{(Q^{-1}(\delta))^{2}+2\epsilon} \right)$.
The following theorem can be found in \cite{LeNyTAC2014}.

\begin{thm}	\label{eq: basic DP mechanism}
Let $\epsilon > 0$, $1 > \delta > 0$. Let $G$ be a dynamic system with $p$ inputs and $q$ outputs. 
Then the mechanism $M(y) = Gy + \nu$, where $\nu$ is a white Gaussian noise (sequence of iid zero-mean 
Gaussian vectors) with $\nu_t \sim \mathcal N(0,\kappa_{\delta,\epsilon}^2 (\Delta_2 G)^2 I_q)$,
is $(\epsilon,\delta)$-differentially private.
\end{thm}

In other words, Theorem \ref{eq: basic DP mechanism} says that we can produce a differentially
private signal by adding white Gaussian noise at the output of a system $G$ processing the sensitive
signal $y$, with covariance matrix $\sigma^2 I_q$, and $\sigma$ proportional to the $\ell_2$-sensitivity of $G$.

\subsection{A First Solution: Input Perturbation Mechanism}
\label{section: input perturbation}

An important property of differential privacy is its \emph{resilience to 
post-processing} \cite{DworkBook}, i.e., applying further computations to an output that is 
differentially private does not degrade the differential privacy guarantee, as long as the original
dataset is not reaccessed for these computations. 
This leads immediately to a first solution for the estimation and control problems stated in the previous 
section, called the \emph{input perturbation mechanism}, which consists in perturbing each measured
signal $y_i$ directly to release differentially private versions of these signals. 

First, note that the memoryless system defined by $(G y)_t = M y_t$, where $M$ is the diagonal matrix
$M = \text{diag} (1/\rho_1, \ldots, 1/\rho_n)$, has the sensibility bound
\[
\Delta G := \sup_{\text{Adj(y,y')}} \|My - My'\|_2 \leq 1
\]
for the adjacency relation \eqref{eq:adjacency}. Hence, by Theorem \ref{eq: basic DP mechanism}, 
releasing the signals $\{y_i/\rho_i + \nu_i\}_{1 \leq i \leq n}$, where each signal $\nu_i$ is 
a white Gaussian noise with covariance matrix $\kappa_{\delta,\epsilon}^2 I_{p_i}$, 
is $(\epsilon,\delta)$-differentially private. Equivalently, using the resilience to 
post-processing to multiply this output by $M^{-1}$, we see that releasing the signals 
$\{\tilde y_i := y_i + \rho_i \nu_i\}_{1 \leq i \leq n}$ is $(\epsilon,\delta)$-differentially private. 
Once these signals are released, applying further processing on them does not impact the 
differential privacy guarantee. Moreover, these signals are of the same form as the outputs of system \eqref{eq:IndividualGeneral}, except for a higher level of (still Gaussian) noise due to the addition 
of the artificial privacy-preserving noise. One can therefore produce an $(\epsilon,\delta)$-differentially 
private estimate $\hat z$ or control signal $u$ discussed in Section \ref{section: estimation and control problems} 
by applying standard Kalman filtering and LQG design techniques to the signals $\tilde y_i$,
see Figure \ref{fig:problem setup}.
An advantage of input perturbation mechanisms is that each agent can release directly the 
differentially private signal $\tilde y_i$, and hence does not need to trust the data aggregator
to enforce the differential privacy property. Moreover, this mechanism has the potentially useful
feature of publishing the individual signals $\tilde y_i$, which could be used for other
purposes than the original estimation or control problem.
Nonetheless, as we discuss in the following sections, input perturbation typically leads 
to a high level of noise and hence performance degradation, which motivates the search 
for better mechanisms.


\section{Differentially Private Kalman Filtering} \label{sec:DP Kalman Filtering}

\begin{figure}
	\centering
	\includegraphics[width=\columnwidth]{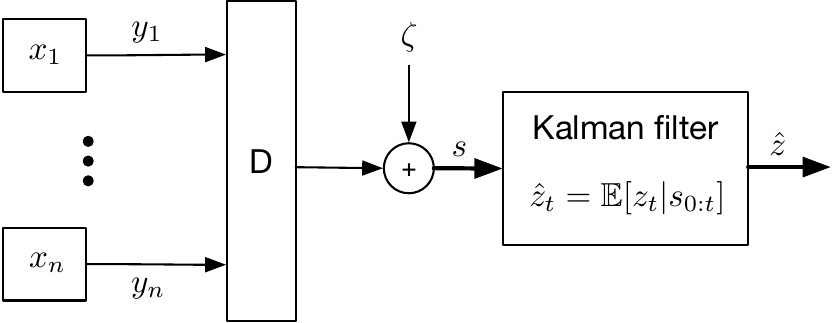}
	\caption{Differentially private Kalman filtering architecture with first-stage aggregation. 
	$D$ is a matrix to design, used to combine the individual measured signals appropriately.
	The signal $s$ is differentially private ($\zeta$ is a privacy-preserving white Gaussian noise signal), 
	and so is the released estimate $\hat z$ by resilience to post-processing.}
	 \label{fig:improvedArchitecture}
\end{figure}

Input perturbation for the differentially private Kalman filtering problem, as discussed 
in Section \ref{section: input perturbation} and represented on Figure \ref{fig:filtering setup}, 
was considered in \cite{LeNyAllerton2012}. 
Here, we show first in Section \ref{sec:Main-results-Scalar} via a simple example
that the performance of this mechanism can be significantly improved by combining 
the individual input signals before adding the privacy preserving noise.
This leads to the two-stage mechanism of Figure \ref{fig:improvedArchitecture},
whose systematic design is discussed in Section \ref{sec:Main-results-General}.

\subsection{A Scalar Example}  \label{sec:Main-results-Scalar}

Consider a scalar homogeneous version of model \eqref{eq:IndividualGeneral} with 
$A_{i,t}=a \in \mathbb R$, $B_{i,t} = b$, $C_{i,t}=c$, $W_{i,t}=\sigma_{w}^{2}$, $V_i=\sigma_{v}^{2}$ 
and $L_{i,t}=1$ (so $z_t = \sum_{i=1}^n x_{i,t}$), 
and assume $\rho_i = \rho$ 
for all $i$ in \eqref{eq:adjacency}.
In other words, $A_t = a I_n$, $B_t = b \mathbf 1_n$, $C_t = c I_n$, 
$W_t = \sigma_w^2 I_n$, $V = \sigma_v^2 I_n$ in \eqref{eq:Global}.
We also let $T \to \infty$ in the problem statement and consider the steady-state MSE
$E_\infty$, see \eqref{eq: asymptotic performance},
as performance measure for a given estimator $\hat z$ of $z$. Moreover in this section 
we consider for simplicity the minimum mean square error (MMSE) estimate 
$\hat z_t$ of $x_t$ given the measurements up to time $t-1$ only, since the corresponding 
MSE is directly obtained by solving an algebraic Riccati equation (ARE).

Let $\alpha := \kappa_{\delta,\epsilon} \, \rho$. Since 
$\sup_{y,y':\text{Adj}(y,y')}\left\Vert y-y'\right\Vert _{2} = \rho$, 
by Theorem \ref{eq: basic DP mechanism} and as explained in Section \ref{section: input perturbation},
releasing the signal $r_{t}=  y_{t}+\nu_{t}$, with white noise $\nu$ such that $\nu_{t}\sim\mathcal{N}(0,\alpha^{2}I_{n})$, 
is $(\epsilon,\delta)$-differentially private for the adjacency relation \eqref{eq:adjacency}. 
The steady-state MSE of a Kalman filter estimating $z$ for the system with dynamics as in \eqref{eq:Global}
and measurements $r$ is obtained by solving a 
scalar ARE, which leads to the following expression 
\begin{equation}
E_\infty^1 =
\frac{n}{2c^{2}}\left(-\beta+\sqrt{\beta^{2}+4(\alpha^{2}
+\sigma_{v}^{2})\sigma_{w}^{2}c^{2}}\right),\label{eq:costnew}
\end{equation}
where $\beta = (1-a^2) (\sigma_v^2+\alpha^2)-c^2 \sigma_w^2$.

Instead of input perturbation, we can use the architecture shown on 
Fig. \ref{fig:improvedArchitecture} with $D=\mathbf{1}_n^T$, 
a $1 \times n$ row vector of ones.
Consider the same adjacency relation (\ref{eq:adjacency}) and denote 
$\eta_{t} = D y_t = \sum_{i=1}^{n}y_{i,t}$, $\theta_{t}=\sum_{i=1}^{n}w_{i,t}$ 
and $\lambda_{t}=\sum_{i=1}^{n}v_{i,t}$.
We have
\begin{align*}
z_{t+1} & = a z_{t}+\theta_{t},\\
\eta_{t} & = c z_{t}+\lambda_{t}.
\end{align*}
Since again $\sup_{y,y':\text{Adj}(y,y')}\left\Vert Dy-Dy'\right\Vert _{2} = \rho$,
releasing the scalar signal $s_{t}= \eta_{t} + \zeta_{t}$, with $\zeta_{t}\sim\mathcal{N}(0,\alpha^{2})$, 
is $(\epsilon,\delta)$-differentially private for the adjacency relation (\ref{eq:adjacency}). 
The MSE of a Kalman filter estimating $z$ from this signal $s$, with the dynamics of the model \eqref{eq:Global},
can again be obtained by solving an ARE, which leads to the following expression
\begin{equation}
E_\infty^2 =
\frac{n}{2c^{2}} \left(-\beta_{(n)}+\sqrt{\beta_{(n)}^{2}+4\left(\frac{\alpha^{2}}{n} +\sigma_{v}^{2}\right) \sigma_{w}^{2} c^{2}}\right), 
\label{eq:cost2new}
\end{equation}
where $\beta_{(n)} = (1-a^2) \left(\sigma_v^2+\frac{\alpha^2}{n} \right)-c^2 \sigma_w^2$.

Comparing \eqref{eq:costnew} and \eqref{eq:cost2new}, we see that the only difference is
the vanishing influence of the privacy preserving noise on $E_\infty^2$ as $n$ increases, with
the term $\alpha^2/n$ replacing $\alpha^2$ in $E_\infty^1$.
For example, if $n=100$, $a=1$, $c=1$, $\rho=50$, $\epsilon=\ln(3)$, $\delta=0.05$, $\sigma_{w}^{2}=0.5$, 
$\sigma_{v}^{2}=0.9$, we obtain $E_\infty^1 \approx 6235$ and $E_\infty^2 \approx 650$.
It is indeed desirable that as the number of agents $n$ increases, differential privacy becomes easier to enforce
and the impact of the privacy requirement on achievable performance decreases, a feature that the architecture 
of Figure \ref{fig:improvedArchitecture} has the potential to achieve. The design of this architecture is discussed 
in the next section for the general filtering problem of Section \ref{sec:statement}.


\subsection{Design of the Two-Stage Mechanism} \label{sec:Main-results-General}

Following Figure \ref{fig:improvedArchitecture}, we construct a differentially private 
estimate $\hat z_t$ of $z_t$ by first multiplying the global signal $y$ with 
a constant matrix 
\begin{equation}	\label{eq: D decomposition}
D=\begin{bmatrix}D_{1} & D_{2} & \ldots & D_{n}\end{bmatrix}, 
\end{equation}
with the matrices $D_i \in \mathbb R^{q\times p_i}$, $1 \leq i \leq n$, to be designed and $q$ to be determined.
Then, we add white Gaussian noise $\zeta$ according to the Gaussian mechanism,
in order to make the signal $s$ differentially private, with
\begin{equation}
s_{t} = D y_{t} + \zeta_{t} = 
D C_{t} x_{t} + D v_{t} + \zeta_{t}, \; 0 \leq t \leq T. 
\label{eq:s-formula}
\end{equation}
Therefore, the role of the matrix $D$ is to combine the individual signals appropriately
before adding the privacy-preserving noise, in order to decrease the overall sensitivity
(see Definition \ref{def: sensitivity}), while preserving enough information for $z$ to
be estimated with sufficient accuracy.
Finally, we construct a causal MMSE estimator $\hat z$ 
of $z$ from $s$, a task for which it is optimal to use a Kalman filter, since the system 
model producing $s$ with the state dynamics of \eqref{eq:Global} is still linear and Gaussian. 
This Kalman filter produces a state estimate $\hat x$ of $x$ and then $\hat z_t = L_t \hat x_t$ 
for all $t$.

Given $D$, for measurement signals $y$ and $y'$ adjacent according to \eqref{eq:adjacency} 
and differing at index $i$, we have
\begin{align*}
\| Dy-Dy' \|_{2} = \| D_{i}y_{i}-D_{i}y_{i}'\|_{2} \leq \rho_i \|D_i\|_2,
\end{align*}
where $\|D_i\|_2$ denotes the maximum singular value of the matrix $D_i$,
and there are adjacent signals $y_i, y_i'$ achieving the bound.
Hence, we can bound the sensitivity of the memoryless system $y \mapsto Dy$ as follows 
\begin{align}
\triangle_{2}D & :=  \sup_{y,y':\text{Adj}(y,y')}\left\Vert Dy-Dy'\right\Vert _{2}\nonumber \\
& =  \underset{1 \leq i \leq n}{\max}\{ \rho_i \|  D_{i} \|_2 \}. \label{eq:sensitivity-2}
\end{align} 
Therefore, from Theorem \ref{eq: basic DP mechanism}, for any matrix $D$, releasing 
$s_{t}=D y_{t}+\zeta_{t}$, with $\zeta_{t}\sim\mathcal{N}(0,(\kappa_{\delta,\epsilon} \, \triangle_{2}D)^{2}I_{q})$,
is $(\epsilon,\delta)$-differentially private for the adjacency relation (\ref{eq:adjacency}). 
The estimate $\hat z$ is then also $(\epsilon,\delta)$-differentially private, since it is obtained
by post-processing $s$, without re-accessing the sensitive signal $y$.

\subsubsection{ \label{sec:SubsectionOptimalProblemFiltering}Input Transformation Optimization}

We can now consider the problem of optimizing the choice of matrix $D$. 
Let $\hat x_t^- = \mathbb E[x_t | s_{0:t-1}]$ and $\hat x_t = \mathbb E[x_t | s_{0:t}]$ be the 
state estimates produced by the Kalman filter of Figure \ref{fig:improvedArchitecture} after the 
prediction step and the measurement update step respectively \cite{Anderson2005}.
Let $\bar \Sigma_{t} = \mathbb E[(x_t - \hat x_t^-) (x_t - \hat x_t^-)^T | s_{0:t-1}]$ and 
$\Sigma_{t}  = \mathbb E[(x_t - \hat x_t) (x_t - \hat x_t)^T | s_{0:t}] $ be the corresponding
error covariance matrices. 
We also denote by $\bar \Sigma_{0} = \text{diag}(\Sigma^-_{1,0}, \ldots \Sigma^-_{n,0})$ 
the covariance matrix for the initial state $x_{0}$. 
For completeness, we recall here the equations of the Kalman filter.  
Given the dynamics \eqref{eq:Global} and the measurement 
equation \eqref{eq:s-formula}, with $\zeta$ a Gaussian noise with covariance matrix 
$(\kappa_{\delta,\epsilon} \, \triangle_{2}D)^2 I_q$, we have for $t \geq 0$ and starting
from $\hat x_0^- := \bar x_0$
\begin{align}
\begin{split}	\label{eq: KF}
\hat x_t &= \hat x_t^- + K^f_t (s_t - DC_t \hat x_t^-),  \\ 
\hat x_{t+1}^- &= A_t \hat x_t + B_t u_t, \quad \quad \text{with }  
\end{split} \\
K^f_t &= \bar \Sigma_t C_t^T D^T (D (C_t \bar \Sigma_t C_t^T + V) D^T 
+ \kappa_{\delta,\epsilon}^2 \Delta_2 D^2 I_q)^{-1}. \nonumber
\end{align}
The error covariance matrices evolve for $t \geq 0$ as
\begin{align*}
\Sigma_{t}^{-1} = \bar \Sigma_{t}^{-1}+C_{t}^T \Pi C_{t}, \;\;  
\bar \Sigma_{t+1} &= A_{t} \Sigma_{t}A_{t}^T+W_{t}, 
\end{align*}
where 
$\Pi = D^T(D V D^T+(\kappa_{\delta,\epsilon} \, \triangle_{2}D)^{2}I_{q})^{-1}D$.
With $z = L_t x_t$ and its estimator $\hat z_t = L_t \hat x_t$, we can rewrite
the MSE $E_T$ in \eqref{eq: MSE performance} as
\[
E_T = \frac{1}{T+1} \sum_{t=0}^{T} \text{Tr} (L_{t} \Sigma_{t} L_{t}^T).
\]
As a result, a matrix $D$ minimizing the MSE can be found by
solving the following optimization problem 
\begin{subequations}
\begin{align}
\min_{q \in \mathbb N, D \in \mathbb R^{q \times p}} & \quad \frac{1}{T+1} \sum_{t=0}^{T} \text{Tr} \left( L_{t} \Sigma_{t} L_{t}^T \right) 	\label{eq: cost function}  \\
\text{s.t.} \; \Sigma_{0}^{-1} &= \; \bar \Sigma_{0}^{-1} + C_{0}^T \Pi C_{0}, 	\label{eq: Sigma_0} \\ 
\Sigma_{t+1}^{-1} &= (A_{t} \Sigma_{t} A_{t}^T+ W_{t})^{-1} + C_{t+1}^T \Pi C_{t+1},  \label{eq: Sigma_t}  \\ 
& \hspace{3cm} 0 \leq t \leq T-1, \nonumber \\
\Pi &= D^{T} (D V D^{T} + \kappa_{\delta,\epsilon}^{2} (\Delta_{2}D)^{2} I_{q})^{-1} D.  \label{eq: SNRt}
\end{align}
\end{subequations}
In the minimization \eqref{eq: cost function}, we have emphasized that finding the first 
dimension $q$ of the matrix $D$ is part of the optimization problem.
Note also that we can write the optimization problem above equivalently as a minimization 
over the variables $D$, $\Pi$, and $\{\Sigma_t\}_{0 \leq t \leq T}$, but the variables 
other than $D$ can be immediately eliminated using the equality 
constraints \eqref{eq: Sigma_0}-\eqref{eq: SNRt}.

With our assumption $V \succ 0$, we obtain an equivalent form for \eqref{eq: SNRt}
by using the matrix inversion lemma 
\begin{align}
\Pi &=  V^{-1} - V^{-1} \left( V^{-1}
+ \frac{D^T D}{(\kappa_{\delta,\epsilon} \Delta_2 D)^2} \right)^{-1} V^{-1},
\label{eq: SNRt alt} 
\end{align}
or, alternatively,
\begin{align}
&\kappa_{\delta,\epsilon}^2 \left[ \left( V - V \Pi V \right)^{-1} - V^{-1} \right] = 
\frac{D^T D}{\Delta_2 D^2}\;\;.
\label{eq: Q matrix construction}
\end{align}

\subsubsection{Semidefinite Programming-based Synthesis}    \label{sec:SubsectionSDPFiltering} 

In this section, we show that the optimization problem \eqref{eq: cost function}-\eqref{eq: SNRt}
can be recast as a semidefinite program (SDP) and hence solved efficiently \cite{BoydLMI1994},
if we impose the following additional constraints on $D$
\begin{equation}	\label{eq: extra constraint}
\Delta_2 D = 1 = \rho_1 \|D_1\|_2 = \ldots = \rho_n \|D_n\|_2.
\end{equation}
First, the following Lemma shows that in fact no loss of performance occurs 
by adding the constraint \eqref{eq: extra constraint} to 
\eqref{eq: cost function}-\eqref{eq: SNRt}, i.e., that this constraint 
is satisfied automatically by some matrix $D^*$ that is optimal 
for \eqref{eq: cost function}-\eqref{eq: SNRt}. 

\begin{lem}	\label{lem: extra constraint}
For any feasible solution 
$D$ of \eqref{eq: Sigma_0}-\eqref{eq: SNRt} that does not satisfy 
\eqref{eq: extra constraint}, there exists a feasible solution that does satisfy this constraint 
and gives a lower or equal cost in \eqref{eq: cost function}. 
In particular, adding the constraint \eqref{eq: extra constraint} to the problem 
\eqref{eq: cost function}-\eqref{eq: SNRt} does not change the value of the
minimum nor the existence of a minimizer.
\end{lem}

\begin{proof}
Consider a matrix $D$ and a corresponding sequence $\{\Sigma_{t}\}_{0 \leq t \leq T}$ defined by the 
iterations \eqref{eq: Sigma_0}-\eqref{eq: SNRt}.
First, rescaling $D$ to $\lambda D$ for any $\lambda \neq 0$ does not impact the constraint 
\eqref{eq: SNRt} (note that $\Delta_2 (\lambda D) = \lambda \Delta_2 D$), and so we can add
the constraint $\Delta_2 D = 1$ without changing the solution of the optimization problem 
\eqref{eq: cost function}-\eqref{eq: SNRt}.

Next, if the other constraints of \eqref{eq: extra constraint} are not satisfied by $D$, construct 
the $p \times p$ matrix $M = D^T D + \text{diag}\left(\left\{ \eta_i I_{p_i}\right\}_{1 \leq i \leq n}\right)$,
with $\eta_i = (\Delta_2 D / \rho_i)^2 - \|D_i\|^2_2$. Since $\eta_i \geq 0$ by \eqref{eq:sensitivity-2},
$M$ is positive semi-definite. The $i^{\text{th}}$ diagonal block of $M$ is
$M_{ii} = D_i^T D_i + [(\Delta_2 D / \rho_i)^2 - \|D_i\|^2_2] I_{p_i}$, which has
maximum eigenvalue $(\Delta_2 D / \rho_i)^2$. Define some matrix $\tilde D$ such that
$\tilde D^T \tilde D = M$ and group the columns of $\tilde D$ as 
$\tilde D = \begin{bmatrix} \tilde D_1 & \ldots \tilde D_n \end{bmatrix}$ as for $D$, so that
$\tilde D_i$ consists of $p_i$ columns. In particular $M_{ii} = \tilde D_i^T \tilde D_i$, so $\tilde D_i^T \tilde D_i$
has maximum eigenvalue $(\Delta_2 D / \rho_i)^2$ and hence $\tilde D_i$ has maximum singular
value $(\Delta_2 D / \rho_i)$. In other words, $\tilde D$ satisfies \eqref{eq: extra constraint}
with a sensibility $\Delta_2 D = \Delta_2 \tilde D$ that is unchanged, and moreover $\tilde D^T \tilde D = M \succeq D^T D$.

Therefore, when we replace $D$ by $\tilde D$, $\Delta_2 D$ in the denominator of \eqref{eq: SNRt alt} remains unchanged,
and moreover 
\begin{align*}
\left( V^{-1} + \frac{\tilde D^{T} \tilde D}{(\kappa_{\delta,\epsilon} \Delta_2 \tilde D)^2} \right)^{-1} 
&\preceq \left( V^{-1} + \frac{D^{T} D}{(\kappa_{\delta,\epsilon} \Delta_2 D)^2} \right)^{-1} 
\end{align*}
hence $\tilde \Pi \succeq \Pi$,
where $\tilde \Pi$ and $\Pi$ are defined according to \eqref{eq: SNRt alt} or equivalently \eqref{eq: SNRt} 
for $\tilde D$ and $D$ respectively. Let $K := \tilde \Pi - \Pi \succeq 0$.
Replacing $\Pi$ by $\tilde \Pi$ in \eqref{eq: Sigma_0}, we obtain a matrix $\tilde \Sigma_{0}$
satisfying $\tilde \Sigma_{0}^{-1} = \Sigma_{0}^{-1} + C_{0}^T K C_{0} \succeq \Sigma_{0}^{-1}$,
so $\tilde \Sigma_{0} \preceq \Sigma_{0}$.
Now if we have two matrices $\tilde \Sigma_{t} \preceq \Sigma_{t}$, and we use these two matrices 
together with $\tilde \Pi$ and $\Pi$ to define $\tilde \Sigma_{t+1}, \Sigma_{t+1}$ according to \eqref{eq: Sigma_t}, 
then immediately 
\begin{align*}
\tilde \Sigma_{t+1}^{-1} &= (A_{ t} \tilde \Sigma_{ t} A_{ t}^{T} 
+ W_{ t})^{-1} + C_{t+1}^T \tilde \Pi C_{t+1} \\
&\succeq (A_{ t} \Sigma_{ t} A_{ t}^{T} + W_{ t})^{-1} + 
C_{t+1}^T \tilde \Pi C_{t+1} \\
&= \Sigma_{t+1}^{-1} + C_{t+1}^T K C_{t+1} \succeq \Sigma^{-1}_{t+1}.
\end{align*}
Therefore, $\tilde \Sigma_{t+1} \preceq \Sigma_{t+1}$.
Hence, by induction, starting from $\tilde D$ we obtain a sequence $\{\tilde \Sigma_{t}\}_t$ such 
that $\tilde \Sigma_{t} \preceq \Sigma_{t}$ for all $t \geq 0$.
This gives a smaller or equal cost
\begin{align*}
\frac{1}{T+1} \sum_{t=0}^{T+1} \text{Tr} (L_{t} \tilde \Sigma_{t} L_{t}^T) \leq \frac{1}{T+1} \sum_{t=0}^{T+1} \text{Tr} (L_{t} \Sigma_{t} L_{t}^T),
\end{align*}
and so the lemma is proved.
\end{proof}

By Lemma \ref{lem: extra constraint}, we can add without loss of optimality the constraints 
\eqref{eq: extra constraint} to \eqref{eq: cost function}-\eqref{eq: SNRt}, which allows us
in the following to recast the problem as an SDP. 
Let $\alpha_i = \kappa_{\delta,\epsilon} \rho_i$, for all $1 \leq i \leq n$.  
Denote $E_{i}=\begin{bmatrix} 0& \ldots & I_{p_{i}} & \ldots & 0\end{bmatrix}^T$ 
the $p \times p_i$ matrix whose elements are zero except for an identity matrix in 
its $i^{\text{th}}$ block. The next lemma converts the constraints 
\eqref{eq: Q matrix construction}-\eqref{eq: extra constraint} to linear 
matrix inequalities. 

\begin{lem} 	\label{lem: Q matrix construction}
If $\Pi, D$ satisfy the constraints \eqref{eq: Q matrix construction}-\eqref{eq: extra constraint},
then $\Pi$ satisfies $V - V \Pi V \succeq 0$ together with the following constraints, 
for all $1 \leq i \leq n$,
\begin{align*}
& \begin{bmatrix} I_{p_{i}}/\alpha_{i}^{2} + V_i^{-1} & E_{i}^T\\
E_{i} & V-V \Pi V
\end{bmatrix}\succeq0, \\
& \textup{and not} \hspace{2mm} \begin{bmatrix} I_{p_{i}}/\alpha_{i}^{2} + V_i^{-1} & E_{i}^T\\
E_{i} & V-V \Pi V
\end{bmatrix}\succ0.
 \end{align*} 
Conversely, if $\Pi$ satisfies these constraints, then there exists a matrix $D$ 
such that $\Pi, D$ satisfy \eqref{eq: Q matrix construction}-\eqref{eq: extra constraint}. 
One such $D$ can be obtained by the factorization of
\begin{equation}	\label{eq: D factorization}
\kappa_{\delta,\epsilon}^2 \left[ \left( V - V \Pi V \right)^{-1} - V^{-1} \right] = D^T D
\end{equation}
(e.g., via singular value decomposition (SVD)) and will then satisfy $\Delta_2 D = 1$.
\end{lem}

\begin{proof}
$V - V \Pi V \succeq 0$ is immediate from \eqref{eq: SNRt alt}, since it is equal to 
$\left( V^{-1} + \frac{D^{T} D}{(\kappa_{\delta,\epsilon} \Delta_2 D)^2} \right)^{-1}$.
Together with \eqref{eq: extra constraint}, the right-hand side 
of \eqref{eq: Q matrix construction} then represents any positive 
semidefinite matrix $M = D^T D$ such that
its diagonal blocks $M_{ii} = D^T_i D_i$ have maximum eigenvalue
equal to $1/\rho_i^2$, since $\|D_i\|_2 = 1/\rho_i$ by \eqref{eq: extra constraint}.
These constraints are equivalent to saying that for all $1 \leq i \leq n$,
\begin{align}
& E_i^{T}
\left[ \left( V - V \Pi V \right)^{-1} - V^{-1} \right] E_i\preceq I_{p_{i}}/\alpha_i^2, \label{LemmaCondition1} \\
& \text{and not} \hspace{1mm} E_i^{T} \left[ \left( V - V \Pi V \right)^{-1} - V^{-1} \right] E_i 
\prec I_{p_{i}}/\alpha_i^2. \label{LemmaCondition2}
\end{align} 
Indeed, this comes from the standard fact that the maximum value $\lambda_{i, max}$ 
of $M_{ii}$ is the smallest $\lambda$ satisfying $M_{ii} \preceq \lambda I_{p_i}$. 
The constraints given in the Lemma are obtained by 
noting that $E_{i}^T V^{-1}E_{i} = V_i^{-1}$
and taking Schur complements in \eqref{LemmaCondition1} and \eqref{LemmaCondition2}.

Note that the fact that the left-hand side of \eqref{eq: Q matrix construction} is 
positive semidefinite is a simple consequence of $V \succeq V - V \Pi V$, hence adding
the constraint $\left( V - V \Pi V \right)^{-1} - V^{-1} \succeq 0$
is unnecessary.
\end{proof}

Next, define the information matrices $\Omega_t = \Sigma_t^{-1}$, for $0 \leq t \leq T$. 
If the matrices $W_{t}$ are invertible, denoting
$\Xi_t = W_{t}^{-1}$ and using the matrix inversion lemma in \eqref{eq: Sigma_t}, one gets
\begin{align} 	\label{eq: Sigma_t - rewrite}
&C_{t+1}^{T} \Pi C_{t+1} - \Omega_{t+1} + \Xi_{t} \nonumber \\
&\;\; - \Xi_{t} A_{t} (\Omega_{t} + A_{t}^{T} \Xi_{t} A_{t})^{-1} A_{t}^{T} \Xi_{t} = 0.
\end{align}

Replacing the equality in \eqref{eq: Sigma_t - rewrite} by $\succeq 0$ and taking a Schur complement,
together with  the inequalities of Lemma \ref{lem: Q matrix construction},
leads to the following SDP with variables 
$\Pi \succeq 0, \{X_{t} \succeq 0,\Omega_{t} \succ 0\}_{0 \leq t \leq T}$

\begin{subequations}
\begin{align}
& \min_{ \Pi \succeq 0, \{X_{t}, \Omega_t\}_{0 \leq t \leq T} } \quad \frac{1}{T+1} \sum_{t=0}^{T} \text{Tr} (X_t) \quad \text{s.t.} \label{eq: SDP cost function} \\
& \begin{bmatrix} X_{t} & L_t \\ L_t^{T} & \Omega_{t} \end{bmatrix} \succeq 0, \;\; 0 \leq t\leq T, \label{eq: SDP slack var} \\
& \Omega_{0} = \bar \Sigma_{0}^{-1} + C_{0}^{T} \Pi C_{0}, \label{eq: SDP constraint 0} \\
& \begin{bmatrix}
C_{t+1}^{T} \Pi C_{t+1} - \Omega_{t+1} + \Xi_{t}  & \Xi_{t} A_{t} \\
A_{t}^{T} \Xi_{t} & \Omega_{t} + A_{t}^{T} \Xi_{t} A_{t}
\end{bmatrix} \succeq 0, \nonumber \\
& 0 \leq t \leq T-1, \label{eq: SDP constraint t} \\
&\begin{bmatrix} I_{p_{i}}/\alpha_{i}^{2} + V_i^{-1} & E_{i}^T\\
E_{i} & V-V \Pi V
\end{bmatrix}\succeq0,\;\; 1 \leq i\leq n. 		\label{eq: SDP constraint Pi} 
\end{align}
\end{subequations}
Here the minimization of the cost \eqref{eq: cost function} has been replaced 
by the minimization of \eqref{eq: SDP cost function}, after introducing the 
slack variable $X_t$ satisfying \eqref{eq: SDP slack var}, or equivalently 
$X_t \succeq L_t \Omega_t^{-1} L_t^{T}$ by taking a Schur complement.
Since we replaced the equality in \eqref{eq: Sigma_t - rewrite} by an inequality, the SDP above
is a relaxation of the original problem \eqref{eq: cost function}-\eqref{eq: SNRt}.
The purpose of the next theorem is to show that this relaxation is tight.
Once an optimal solution for this SDP is obtained, we recover an optimal matrix $D$
from $\Pi$ by the factorization \eqref{eq: D factorization}.

\begin{thm} \label{thm:FinalTheoremFiltering}
Let $\Pi^* \succeq 0$,$\{ X^*_{t} \succeq 0,\Omega^*_{t} \succ 0\}_{0 \leq t \leq T}$ 
be an optimal solution for \eqref{eq: SDP cost function}-\eqref{eq: SDP constraint Pi}.
Suppose that for some $0 \leq t \leq T$, we have $L_t(\Omega^*_{t})^{-1}C_t^T \neq 0$.
Let $D^*$ be a matrix obtained from $\Pi^*$ by the factorization \eqref{eq: D factorization}.
Then $D^*$ is an optimal solution for \eqref{eq: cost function}-\eqref{eq: SNRt}, which moreover
satisfies $\|D_i^*\|_2 = 1/\rho_i$ for $1 \leq i \leq n$, with the decomposition \eqref{eq: D decomposition}.
The corresponding optimal covariance matrices $\{\Sigma_t^*\}_{0 \leq t \leq T}$ for the
Kalman filter can be computed using the equations \eqref{eq: Sigma_0}-\eqref{eq: SNRt}.
Finally, the optimal costs of \eqref{eq: cost function}-\eqref{eq: SNRt} and 
\eqref{eq: SDP cost function}-\eqref{eq: SDP constraint Pi} are equal, i.e., the
SDP relaxation is tight.
\end{thm}
 
\begin{rem}
Even though the condition $L_t(\Omega^*_{t})^{-1}C_t^T \neq 0$ introduced to guarantee 
the possibility of constructing the matrix $D$ in the proof is not an explicit condition expressed 
directly in term of the problem parameters, it appears to be a weak requirement in practice. 
\end{rem}

\begin{proof}
Consider $\Pi^{*}$, $\{X_{t}^{*},\Omega_{t}^{*}\}_{0 \leq t \leq T}$ an optimal solution 
of the SDP \eqref{eq: SDP cost function}-\eqref{eq: SDP constraint Pi}. 
As explained in the proof of Lemma \eqref{lem: Q matrix construction}, the constraint \eqref{eq: SDP constraint Pi}
is equivalent to
\begin{align*}
\alpha_i^2 E_i^{T} \left[ \left( V - V \Pi^* V \right)^{-1} - V^{-1} \right] E_i\preceq I_{p_{i}}.
\end{align*}
We show that we cannot have 
\[
\alpha_i^2 E_i^{T} \left[ \left( V - V \Pi^* V \right)^{-1} - V^{-1} \right] E_i \prec I_{p_{i}}.
\] 
Indeed, otherwise there exists $\eta > 0$
such that the matrix $\tilde \Pi = \Pi^* + \eta I_{p}$ still satisfies \eqref{eq: SDP constraint Pi}.
Using this matrix $\tilde \Pi$ in \eqref{eq: SDP constraint 0}, we obtain a matrix
$\tilde \Omega_0 = \Omega_0^* + \eta C_0^T C_0$ feasible for \eqref{eq: SDP constraint 0}.
Now define $\tilde \Omega_1 = \Omega_1^* + \eta C_1^T C_1$. One can immediately check that 
$\tilde \Pi, \tilde \Omega_0$ and $\tilde \Omega_1$ satisfy \eqref{eq: SDP constraint t}
for $t = 0$, using the fact that $\Omega_0^*, \Omega_1^*, \Pi^*$ are feasible and that
$C_0^T C_0 \succeq 0$. Similarly the matrices $\tilde \Omega_{t} = \Omega_t^* + \eta C_t^T C_t$
are feasible in \eqref{eq: SDP constraint t} for all $0 \leq t \leq T$.
Now in \eqref{eq: SDP slack var}, taking a Schur complement, we obtain that the
matrices $\tilde X_t = L_t \tilde \Omega_t^{-1} L_t^T$ are feasible. By the matrix
inversion lemma we can write
\[
\tilde X_t = X^*_t - L_t (\Omega^*_t)^{-1} C_t^{T} K_t C_t (\Omega^*_t)^{-1} L_t^T.
\]
for some matrices $K_t \succ 0$.
These matrices $\tilde X_t$ give a cost 
$\frac{1}{T+1} \sum_{t=0}^T \text{Tr}(X^*_t) - \|L_t (\Omega^*_t)^{-1} C_t^{T} K_t^{1/2}\|_F$,
which is a strict improvement over the assumed optimal solution as soon as one matrix 
$L_t (\Omega^*_t)^{-1} C_t^{T}$ is not zero (since the $K_t$'s are invertible).
Hence, we have a contradiction and so we cannot have 
$\alpha_i^2 E_i^{T} \left[ \left( V - V \Pi^* V \right)^{-1} - V^{-1} \right] E_i \prec I_{p_{i}}$. 
We can then apply Lemma \ref{lem: Q matrix construction} and construct a matrix $D^*$ from $\Pi^*$ 
as in \eqref{eq: D factorization}, so that the pair $\Pi^*$, $D^*$ satisfies 
\eqref{eq: Q matrix construction}-\eqref{eq: extra constraint}.

Let $\mathcal V^*$ be the optimum value of \eqref{eq: SDP cost function}-\eqref{eq: SDP constraint Pi},
and $V^*$ that of \eqref{eq: cost function}-\eqref{eq: SNRt}. First, $\mathcal V^* \leq V^*$ since the
constraints of the original problem have been relaxed to obtain the SDP.
We now show how to construct a sequence $\{ \Sigma_t^* \}_{0 \leq t \leq T}$, which together
with $\Pi^*$ satisfy the constraints of \eqref{eq: cost function}-\eqref{eq: SNRt}
and achieve the cost $\mathcal V^*$, thereby proving the remaining claims of the theorem.
Note that since $\Omega^*_{t} + A_{t}^{T} \Xi_{t} A_{t} \succ 0$, \eqref{eq: SDP constraint t} 
is equivalent to $\mathcal R_{t}(\Omega^*_{t}, \Omega^*_{t+1}) \succeq 0$, where
\begin{align*}
\mathcal R_{t}(\Omega_t, \Omega_{t+1}) :=& \, C_{t+1}^{T} \Pi^* C_{t+1} - \Omega_{t+1} + \Xi_{t} \\
&- \Xi_{t} A_{t} (\Omega_t + A_{t}^{T} \Xi_{t} A_{t})^{-1} A_{t}^{T} \Xi_{t}.
\end{align*}

First, we take $\Sigma_0^* = (\Omega_0^*)^{-1}$. 
If $\mathcal R_{t}(\Omega^*_{t}, \Omega^*_{t+1}) = 0$ for all $0 \leq t \leq T-1$,
then the matrices $\Omega^*_t$ satisfy \eqref{eq: Sigma_t - rewrite}
and we can take $\Sigma_t^* = (\Omega_t^*)^{-1}$ for all $t$, since these
matrices satisfy the equivalent condition \eqref{eq: Sigma_t}.
Otherwise, let $\tilde t$ be the first time index such that 
$\mathcal R_{\tilde t}(\Omega^*_{\tilde t}, \Omega^*_{\tilde t+1})$ is not zero.
For $t \leq \tilde t$, we take $\Sigma_t^* = (\Omega_t^*)^{-1}$ and so in particular
we have $\mathcal R_{\tilde t}((\Sigma_{\tilde t}^*)^{-1}, \Omega^*_{\tilde t+1}) \succeq 0$
and not zero.
Consider the matrix 
$\tilde \Omega_{\tilde t+1} = \Omega^*_{\tilde t+1} + \mathcal R_{\tilde t}(\Omega^*_{\tilde t}, \Omega^*_{\tilde t+1})$, 
which then satisfies $\mathcal R_{\tilde t}(\Omega^*_{\tilde t}, \tilde \Omega_{\tilde t+1}) = 0$ 
by definition. We set $\Sigma_{\tilde t+1}^* = \tilde \Omega_{\tilde t+1}^{-1}$.

Now note that we again have 
$\mathcal R_{\tilde t+1}((\Sigma_{\tilde t+1}^*)^{-1}, \tilde \Omega_{\tilde t+2}) \succeq 0$,
by verifying that \eqref{eq: SDP constraint t} is satisfied at $t+1$, using the fact that
$(\Sigma_{\tilde t+1}^*)^{-1} = \tilde \Omega_{\tilde t+1} \succeq \Omega^*_{\tilde t+1}$.
From here, we can proceed by induction, assuming that $\Sigma_{0}^*$, \ldots, $\Sigma_{t}^*$ 
are set and taking
\begin{equation}	\label{eq: Sigma construction}
\Sigma_{t+1}^* = (\tilde \Omega_{t+1})^{-1} := (\Omega_{t+1}^*+\mathcal R_{t}((\Sigma_{t}^*)^{-1}, \Omega^*_{t+1}))^{-1},
\end{equation}
which reduces to $(\Omega_{t+1}^*)^{-1}$ if $\mathcal R_{t}((\Sigma_{t}^*)^{-1}, \Omega^*_{t+1}) = 0$.

The procedure above provides matrices $\Pi^{*}$, $\{\Sigma_t^* \}_{0 \leq t \leq T}$ 
satisfying the constraints of the original program \eqref{eq: cost function}-\eqref{eq: SNRt}.
By construction, we have $(\Sigma_t^*)^{-1} \succeq \Omega_t^*$ and the matrices $(\Sigma_t^*)^{-1}$ 
also satisfy \eqref{eq: SDP constraint t}. Therefore, replacing, for each $0 \leq t \leq T$, $\Omega_t^*$ by 
$(\Sigma^*_t)^{-1}$ and $X_t^*$ by $L_t \Sigma_t L_t^T$ in the solution 
of \eqref{eq: SDP cost function}-\eqref{eq: SDP constraint Pi} that we started with gives a cost
$\mathcal V \leq \mathcal V^*$ for the SDP, hence $\mathcal V = \mathcal V^*$ by optimality of $\mathcal V^*$.
But this cost $\mathcal V$ is also equal to the cost $\frac{1}{T+1} \sum_{t=0}^{T} \text{Tr} (L_{t} \Sigma_{t}^* L_{t}^T)$ 
of \eqref{eq: cost function}. Hence, we have shown that
\eqref{eq: cost function}-\eqref{eq: SNRt} and
\eqref{eq: SDP cost function}-\eqref{eq: SDP constraint Pi} have
the same minimum value, and constructed an optimal solution $D^*$, $\Pi^*$, $\{\Sigma_t^* \}_{0 \leq t \leq T}$
to \eqref{eq: cost function}-\eqref{eq: SNRt} achieving this value.
\end{proof}

\subsection{Stationary problem} \label{sec:Stationary problems}

In the stationary case with $T \to \infty$ and the model 
\eqref{eq:Global}-\eqref{eq:Global measurements} now assumed time-invariant 
and detectable, we wish to find a signal aggregation matrix $D$ followed 
by a time-invariant Kalman filter to minimize the steady-state MSE $E_\infty$.
This can be done by solving the following SDP with variables $\Pi \succeq 0$, 
$X \succeq 0$, $\Omega \succ 0$
\begin{subequations}
\begin{align}
& \min_{\Pi \succeq 0, X, \Omega} \quad   \text{Tr} (X) \label{eq: SDP cost function-Invariant-1} \quad \text{s.t.} \\
& \begin{bmatrix} X & L \\ L^{T} & \Omega \end{bmatrix} \succeq 0, \\
& \begin{bmatrix}
C^{T} \Pi C - \Omega + \Xi  & \Xi A \\
A^{T} \Xi & \Omega + A^{T} \Xi A
\end{bmatrix} \succeq 0, \\
&\begin{bmatrix}I_{p_{i}}/\alpha_i^2+ V_i^{-1} & E_{i}^T\\
E_{i} & V-V \Pi V
\end{bmatrix}\succeq0,\;\; 1 \leq i\leq n. \label{eq: SDP cost function-Invariant-5}	
\end{align}
\end{subequations}

Compared to \eqref{eq: SDP cost function}-\eqref{eq: SDP constraint Pi}, this SDP is of much smaller
size, due to the fact that the transient behavior is neglected in the performance measure.
The proof of the following theorem is similar to that of Theorem \ref{thm:FinalTheoremFiltering}.
\begin{thm} \label{thm:FinalTheoremFiltering - steady state}
Let $\Pi^* \succeq 0$, $X^* \succeq 0$, $\Omega^* \succ 0$ be an optimal solution 
for \eqref{eq: SDP cost function-Invariant-1}-\eqref{eq: SDP cost function-Invariant-5}.
Suppose that we have $L (\Omega^*)^{-1} C^T \neq 0$.
Let $D^*$ be a matrix obtained from $\Pi^*$ by the factorization \eqref{eq: D factorization}.
Then $D^*$ minimizes the steady-state MSE $E_\infty$ among all possible matrices $D$ 
introduced as in Figure \ref{fig:improvedArchitecture}, and the corresponding value of $E_\infty$
is equal to the optimal value of the SDP.
\end{thm}

Note that given the optimum matrix $D^*$ and corresponding $\Pi^*$, an alternative way of 
computing $E_\infty$ is by solving an ARE to obtain the steady-state prediction error 
covariance matrix $\bar \Sigma_\infty$ for $\hat x^-$, then compute the steady-state error 
covariance matrix $\Sigma_\infty = (\bar \Sigma_\infty^{-1} + C^T \Pi^* C)^{-1}$ for 
the estimator $\hat x$, and finally $E_\infty = \Tr(L \Sigma_\infty L^T)$.


\subsection{Syndromic Surveillance Example} \label{sec:Examples}

To illustrate the differentially private filtering methodology, including issues related
to the choice of model and adjacency relation \eqref{eq:adjacency}, we discuss in this
section an example motivated by the analysis of epidemiological data.
Consider a scenario in which Public Health Services (PHS) must publish for a population
infected by a disease the number $I_t$ of infectious people, i.e., those who have the disease 
and are able to infect others. 
PHS use privacy-sensitive data collected from $n=12$ hospitals, 
with each hospital $i$ recording the number $I_{i,t}$ of infectious people in its area,
as well as the number $R_{i,t}$ of recovered people, i.e, those who were infected
by the disease and are now immune.
For each area $i$, these numbers are assumed to follow 
a discrete-time SEIR epidemiological model for the specific disease \cite{Dukic2012,DegueIJC2019},
written here first without process noise,
obtained by discretizing a classical continuous-time model \cite{Hethcote2000} using 
a forward Euler discretization \cite{HU2012}
\begin{align}
\label{eq:modelSEIR}
\begin{split}
S_{i,t+1} & =S_{i,t}-\beta_i S_{i,t}I_{i,t}/N_i, \\
E_{i,t+1} & =(1-\tau_i)E_{i,t}+\beta_i S_{i,t}I_{i,t}/N_i,\\
I_{i,t+1} & =(1-\vartheta_i)I_{i,t}+\tau_i E_{i,t}, \\
R_{i,t+1} & =R_{i,t}+\vartheta_i I_{i,t}, 
\end{split}
\end{align}
where $S_i$ represents the number of susceptible people, i.e., those who are not infected 
but could become infected, the number of exposed people, i.e., those infected but not yet 
able to infect others, and $N_i$ the total number of people.
The parameters $\tau_i$, $\beta_i$ and $\vartheta_i$ represent the transition rates from 
one disease stage to the next. 

For each hospital $i$'s area, $N_i$ in \eqref{eq:modelSEIR} is assumed constant  
for the time interval of interest.
Moreover, let us make an approximation that this period  
is short enough or the disease at an early-enough stage so that
$S_{i,t}$ can also be assumed approximately constant, equal to $S_{i,0}$.
Then, the remaining states 
$\eta_{i,t}=[E_{i,t}, I_{i,t}, R_{i,t}]^T \in \mathbb{R}^{3}$
evolve as a linear system of the form
\begin{align}   \label{eq: preliminary linear model}
\eta_{i,t+1} &= \mathcal{A}_{i} \, \eta_{i,t}+ \varphi_{i,t}, \;\; 0 \leq t \leq T-1,
\end{align}
where 
\begin{align}   
\mathcal{A}_{i}= & \begin{bmatrix}1-\tau_i & \beta_i S_{i,0}/N_i & 0\\
\tau_i & 1-\vartheta_i & 0 \\
0 & \vartheta_i & 1
\end{bmatrix},
\end{align}
and we introduced the white Gaussian process noise $\varphi_{i,t} \sim \mathcal{N}(0,\Phi_{i})$
in the model, with covariance matrices $\Phi_{i} \succ 0$, for $1 \leq i \leq 12$. 
Now, consider the problem of choosing the level $\rho_i$ in the adjacency
relation \eqref{eq:adjacency} to provide a meaningful privacy guarantee to the
patients. If we assume that the measurement for hospital $i$ in our model is 
$y_{i,t} = [I_{i,t}, R_{i,t}]^T$, then once a person becomes sick, he or she is 
counted at each period either in the signal $I_{i}$ or, after recovery, in the 
signal $R_i$. As a result, the impact of a single individual on the measurement $y_i$ 
could be quite large, proportional to the time horizon $T$, requiring in turn a large 
value of $\rho_i$ to provide strong privacy guarantees, and hence a high level
of noise.

A practical solution to this issue is to not continuously record the same individuals, 
but simply count at each time period $t$ the number of newly infectious individuals,
i.e., $y_{i,t}^{(1)} = I_{i,t}-I_{i,t-1}$, as well as the number of newly recovered individuals,
$y_{i,t}^{(2)} = R_{i,t}-R_{i,t-1}$. Such a measurement model is much more beneficial from
a privacy point of view. Indeed, assuming that a given individual can only become
sick once, he or she can affect $y_{i,t}^{(1)}$ by $\pm 1$ for at most $2$ periods $t$
(when the person becomes sick and recovers), and $y_{i,t}^{(2)}$ for at most one period by $1$.
As a result, one can take $\rho_i = \sqrt{3}$ in \eqref{eq:adjacency} to provide a strong
privacy guarantee, i.e., insensitivity of the published output to the complete record of 
a single individual. The measurement noise $v_{i,t}$ in the model represents counting
errors, e.g., due to people not being diagnosed by the hospital.

To obtain a dynamic model compatible with the measurements $y_{i,t} = [y_{i,t}^{(1)}, y_{i,t}^{(2)}]^T$,
define for each hospital $i$ the $4$-dimensional state 
$x_{i,t} = [I_{i,t-1}, R_{i,t} - R_{i,t-1}, E_{i,t}, I_{i,t}]^T$.
These states $x_{i,t}$ evolve as \eqref{eq:IndividualGeneral} with  
\begin{align}   \label{eq: delayed SEIR model}
A_{i}= & \begin{bmatrix}0 & 0& 0 & 1 \\
0 & 0 & 0 & \vartheta_i \\
0 & 0 & 1-\tau_i & \beta_i S_{i,0}/N_i \\
0 & 0 & \tau_i & 1-\vartheta_i
\end{bmatrix}, \; 
B_{i}= \begin{bmatrix}0\\0 \\0 \\0\end{bmatrix}
\end{align}
and the noise $w_{i,t} = [\varpi_{i,t}, \varphi_{1,i,t}, \varphi_{2,i,t}, \varphi_{3,i,t}]^T$
includes the components of $\varphi_i$ introduced in \eqref{eq: preliminary linear model}
as well as a small independent Gaussian noise $\varpi_i$ with small variance $\sigma^2$, 
added so that the covariance matrices of $W_i = \text{diag}(\sigma^2,\Phi_i)$ are invertible 
as required by our algorithms (ideally $\varpi_{i,t}$ would be $0$ since the first line
of $A_i$ corresponds to the delay in the model).
The measurement matrices are immediately
\begin{align*}
C_{i}=\begin{bmatrix}-1 & 0 & 0 & 1\\
0 & 1 & 0 & 0
\end{bmatrix}, \; 1 \leq i \leq 12,
\end{align*}
and one can then verify that the model is observable.
Let us assume the following values for the parameters in \eqref{eq: delayed SEIR model}
\begin{align*}
&\tau_{i}=0.2, \; \beta_i S_{i,0}/N_i = 0.5, \; \vartheta_i=0.1, \text{ for } \; 1 \leq i \leq 3 \\
&\tau_{i}=0.3, \; \beta_i S_{i,0}/N_i = 0.3, \; \vartheta_i=0.5, \text{ for } \; 4 \leq i \leq 6 \\
&\tau_{i}=0.5, \; \beta_i S_{i,0}/N_i = 0.7, \; \vartheta_i=0.15,  \text{ for } \; 7 \leq i \leq 9 \\
&\tau_{i}=0.7, \; \beta_i S_{i,0}/N_i = 0.6, \; \vartheta_i=0.3, \text{ for } \; 10 \leq i \leq 12.
\end{align*}
Moreover, assume for all $1 \leq i \leq 12$
\begin{align*}
V_{i}=0.4 \, I_{2}, & \, 
\Phi_i = \begin{bmatrix}0.3 & -0.15 & 0\\
-0.15 & 0.3 & -0.15 \\
0 & -0.15 & 0.3
\end{bmatrix}.
\end{align*}
The goal of 
the surveillance system is to continuously release,
at each period $t$, an estimate of the quantity 
\begin{align*}
z_{t}= & \sum_{i=1}^{n} \left( \begin{bmatrix}0 & 0 & 0 & 1\end{bmatrix}\times x_{i,t} \right).
\end{align*}

\begin{figure}
\begin{centering}
\includegraphics[width=8cm]{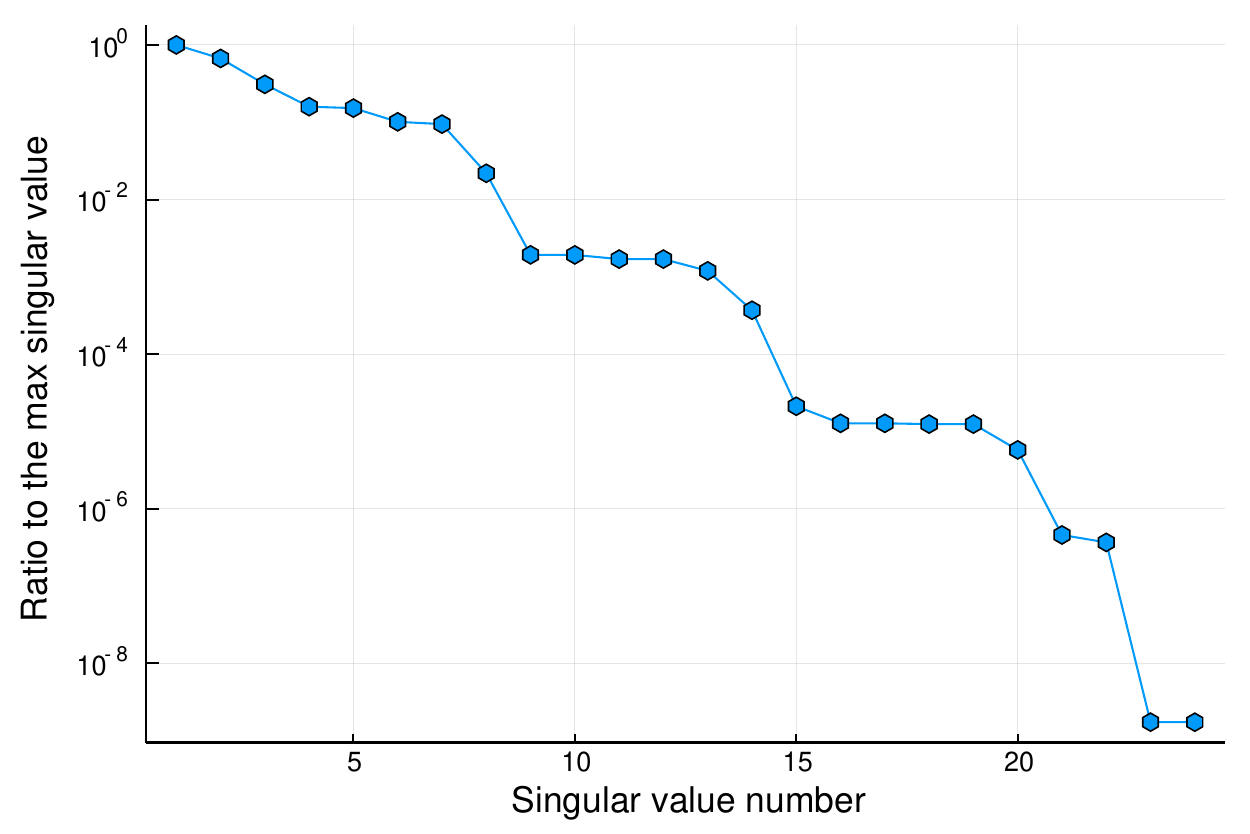}
\par\end{centering}
\caption{Ratio $\sigma_i/\sigma_{\max}$ for the singular values of $M^*$ (on a
logarithmic scale).}
\label{fig:singularvaluesD} 
\end{figure}

Let us set the privacy parameters to $\delta=0.02$ and $\epsilon=\ln(3)$ for example
(and $\rho_i$ to $\sqrt{3}$ as discussed above).
We design the two-stage architecture of Figure \ref{fig:improvedArchitecture} by first 
solving the stationary optimization 
problem \eqref{eq: SDP cost function-Invariant-1}-\eqref{eq: SDP cost function-Invariant-5},
which provides an optimal matrix $\Pi^*$.
Recall that the matrix $D$ can be then obtained from the factorization \eqref{eq: D factorization}.
The number of rows $q$ of $D$ is then equal to the rank of the matrix
$M^* := \kappa_{\delta,\epsilon}^2 \left[ \left( V - V \Pi^* V \right)^{-1} - V^{-1} \right]$.
Due to the numerical procedures, this rank will typically be maximal 
(here equal to $24$). However, we plot on Fig. \ref{fig:singularvaluesD} the 
ratios $\sigma_i/\sigma_{\max}$ of the singular values of $M^*$, with $\sigma_{\max}$
the maximum singular value. We see that if we select for example only the singular
values $\sigma_i \geq 10^{-4} \sigma_{\max}$ in the SVD of $M^*$ and set the smaller
ones manually to $0$, we obtain a matrix of rank $14$, hence a matrix $D$ with $14$ rows
instead of $24$. We then verify (by solving an algebraic Riccati equation) that the 
performance of the steady-state Kalman filter is left virtually unchanged by this 
truncation, with a steady-state MSE of about $160$, i.e., a root mean square error (RMSE) 
of $12.65$ for the estimate of the number $I_t$ of infectious people. 
In contrast, the input perturbation mechanism (i.e., taking $D = I$) gives a steady-state 
MSE of $777$ or RMSE on $I_t$ of $27.87$. 
Reducing the number of rows of $D$ is also beneficial for example in terms 
of processing complexity of the Kalman filter, which now has fewer inputs.

\begin{figure}
\begin{centering}
\includegraphics[width=8cm]{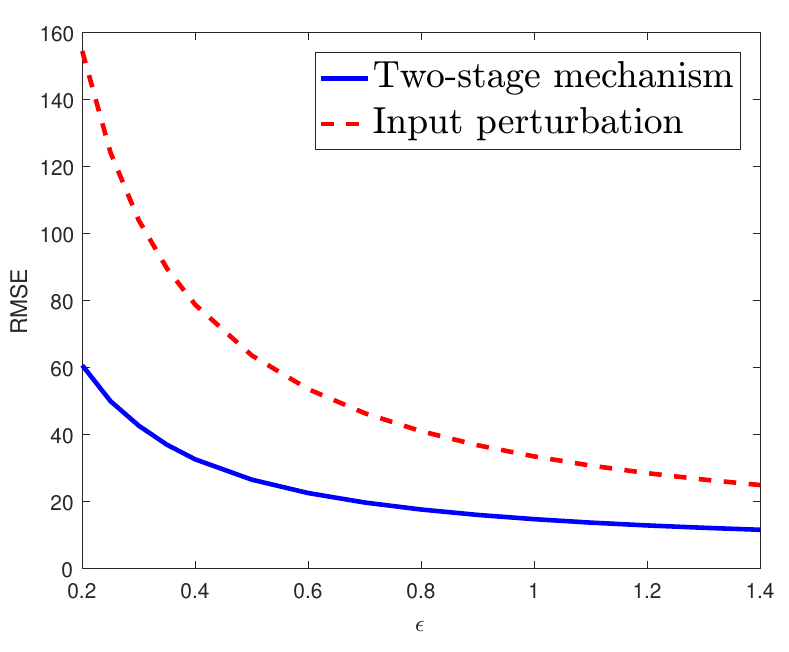}
\par\end{centering}
\caption{\label{fig:RMSE} Steady-state RMSE of the total infectious 
population estimate for the two-stage and input perturbation mechanisms, 
as a function of the privacy parameter $\epsilon$ (here $\delta = 0.01$).
}
\end{figure}

For $\delta=0.01$, we compare on Fig. \ref{fig:RMSE} the steady-state RMSE 
of the two-stage mechanism and the input perturbation architecture for 
different values of the privacy parameter $\epsilon$. One can see 
that by aggregating the input signals, we obtain a much better performance, 
especially in the high-privacy regime (when $\epsilon$ is small).
Other measures of performance could also be of interest for the final filtering
architecture, such as the convergence time of the estimates. 
For illustration purposes, Fig. \ref{fig:Kalman filter} shows sample paths 
of differentially private estimates both for the two-stage mechanism and 
for the input perturbation mechanism.

\begin{figure}
\centering
\includegraphics[width=\linewidth]{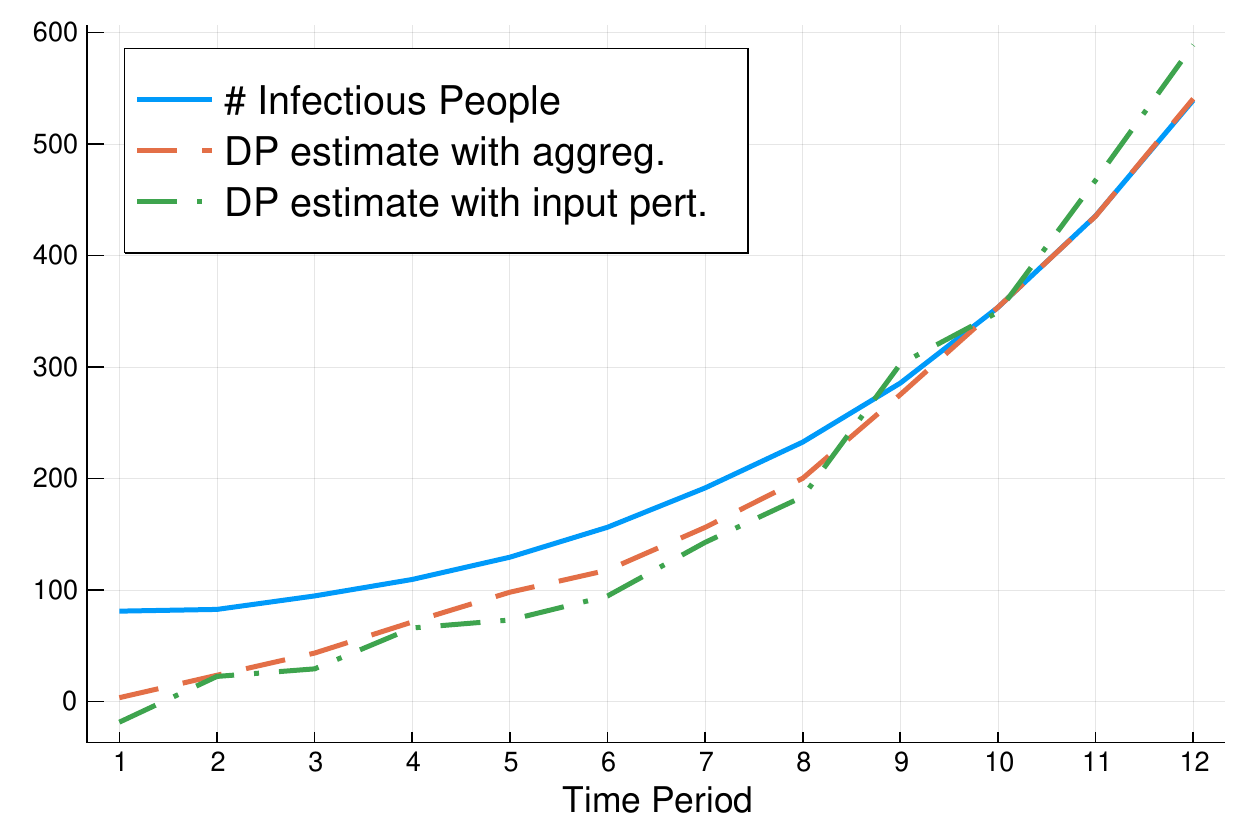}
\caption{Sample paths of differentially private estimates 
(with $\delta=0.01, \; \epsilon=\ln(3)$) of the total infectious population, 
for the two-stage and input perturbation mechanisms. Both estimates 
are initialized with $\bar x_0 = 0$ and we show here the values after 
the measurement update step of the Kalman filter.}
\label{fig:Kalman filter}
\end{figure}


\section{Differentially Private LQG Control}
\label{sec:ProblemFormulationControl}

\begin{figure}
\centering
\includegraphics[width=\linewidth]{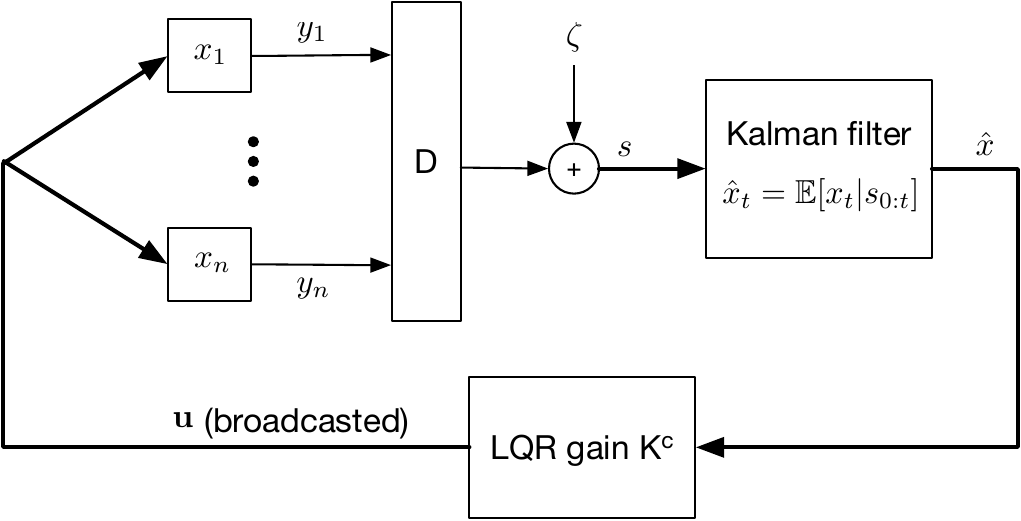}
\caption{Differentially private LQG control architecture with first-stage signal aggregation.}
\label{fig: DP LQG control - aggregation}
\end{figure}

We now turn to the LQG control problem introduced at the end of 
Section \ref{section: estimation and control problems}. For 
concreteness, we assume here that the control input $u_t$ at time $t$ 
can depend on the measurements $y_{0:t}$ up to time $t$. 
It is straightforward to adapt the discussion to the case where 
only $y_{0:t-1}$ are available to compute $u_t$. By the separation 
principle \cite[Chapter 8]{Astrom:book2012:stochasticControl} for
the standard LQG control problem (i.e., with no privacy constraint), 
the optimal control law for the system \eqref{eq:Global}-\eqref{eq:Global measurements} 
and quadratic cost \eqref{eq: cost function LQG} is of the form 
$u_t = K^c_t \hat x_t$, where: i) $\hat x_t = \mathbb E[x_t | y_{0:t}]$ is 
the MMSE estimator, computed by the Kalman filter \eqref{eq: KF} independently 
of the design of the optimal control law; and ii) $K^c_t$ is the optimal 
gain for the deterministic linear quadratic regulator (LQR) problem, i.e., 
assuming that $w = 0$ in \eqref{eq:Global} and $C = I_n$, $v = 0$ 
in \eqref{eq:Global measurements}. 
In particular, since the sequence of control gains $K^c_t$ can be precomputed, 
the LQG control problem is similar to the filtering problem considered in the 
previous section, with the desired published output $\hat z_t = L_t \hat x_t$ 
simply replaced by $u_t = K^c_t \hat x_t$. This motivates the architecture proposed 
on Figure \ref{fig: DP LQG control - aggregation} for differentially private 
LQG control, which, compared to Figure \ref{fig:control setup}, aggregates the 
measured signals $y_i$ before adding the privacy-preserving noise. 
Essentially, the only difference with the Kalman filtering problem is that the 
performance is measured by \eqref{eq: cost function LQG} instead of the 
MSE \eqref{eq: MSE performance}, so that the cost function in the optimization 
problem for the matrix $D$ needs to be changed.
The following theorem summarizes the discussion above and the classical results (see for example
\cite[Chapter 8]{Astrom:book2012:stochasticControl}) that allow us to formulate in the following 
an efficiently solvable optimization problem for the choice of aggregation matrix $D$ on 
Figure \ref{fig: DP LQG control - aggregation}. 

\begin{thm}
Given a choice of matrix $D$ for the differentially private LQG control architecture of 
Figure \ref{fig: DP LQG control - aggregation}, the control law $u_t(y_{0:t})$, $t \geq 0$, 
minimizing the cost function \eqref{eq: cost function LQG} takes the form
\[
u_t = K^c_t \hat x_t,
\]
where $\hat x_t$ is computed by the Kalman filter \eqref{eq: KF} and the gains 
$K^c_t$  are precomputed independently of the filtering problem as 
\[
K^c_t := - (R_t + B_t^T P_{t+1} B_t)^{-1} B_t^T P_{t+1} A_t
\]
with the matrices $P_t \succeq 0$ given by $P_T = Q_T$ and the backward Riccati 
difference equation
\begin{align}	
P_t = &A_t^T P_{t+1} A_t + Q_t \nonumber \\
&- A_t P_{t+1} B_t (R_t + B_t^T P_{t+1} B_t)^{-1} B_t^T P_{t+1} A_t.  \label{eq: Riccati control}
\end{align}
Moreover, the optimal objective \eqref{eq: cost function LQG} corresponding to this 
control law can be written
\begin{equation}	\label{eq: total LQG cost}
J_T = J^c_T + J^f_T(D)
\end{equation}
where 
\[
J^c_T = \frac{1}{T+1} \left( \bar x_0^T P_0 \bar x_0 + \Tr(P_0 \bar \Sigma_0) 
+ \sum_{t=1}^T \Tr(P_t W_t) \right)
\]
is a term independent of $D$ and
\begin{align}
J_T^f(D) &= \frac{1}{T+1} \sum_{t=0}^{T-1} \Tr(N_t \Sigma_t), \;\; \text{ with } 
\label{eq: cost to optimize LQG} \\
N_t &= Q_t + A_t^T P_{t+1} A_t - P_t, \;\; 0 \leq t \leq T-1,  \label{eq: def. of N}
\end{align}
and the matrices $\Sigma_t$ for $0 \leq t \leq T-1$ defined by \eqref{eq: Sigma_0}-\eqref{eq: SNRt}.
\end{thm}

Note that the dependence on $D$ in \eqref{eq: cost to optimize LQG} is due to the fact 
that the error covariance matrices 
$\Sigma_t$ depend on $D$ via \eqref{eq: Sigma_0}-\eqref{eq: SNRt}.
Moreover, from \eqref{eq: Riccati control} we see that 
$N_t$ defined in \eqref{eq: def. of N} is positive semidefinite. 
Hence, we can define for all $0 \leq t \leq T-1$ matrices $L_t$ such that $N_t = L_t^T L_t$, 
and $L_T = 0$, to rewrite the cost \eqref{eq: cost to optimize LQG} as $\frac{1}{T+1} \sum_{t=0}^{T} \Tr(L_t \Sigma_t L_t^T)$.
Minimizing this cost over the matrices $D$ with the relations 
\eqref{eq: Sigma_0}-\eqref{eq: SNRt} leads to an optimal 
aggregation matrix $D$ for the architecture of Figure  \ref{fig: DP LQG control - aggregation}. 
The reformulation of this optimization problem as an SDP then follows exactly from the 
same argument as in Section \ref{sec:SubsectionSDPFiltering},
which led to Theorem \ref{thm:FinalTheoremFiltering}. In other words, we have 
the following result.

\begin{prop}	\label{prop: LQG result}
Let $L_t$, $0 \leq t \leq T-1,$ be any matrices obtained from the factorization
\[
L_t^T L_t = N_t, 0 \leq t \leq T-1,
\]
with $N_t$ defined by \eqref{eq: def. of N}, and let $L_T = 0$.
Let $\Pi^* \succeq 0$,$\{ X^*_{t} \succeq 0,\Omega^*_{t} \succ 0\}_{0 \leq t \leq T}$ 
be an optimal solution for \eqref{eq: SDP cost function}-\eqref{eq: SDP constraint Pi}
with this choice of matrices $L_t$.
Suppose that for some $0 \leq t \leq T$, we have $L_t(\Omega^*_{t})^{-1}C_t^T \neq 0$.
Let $D^*$ be a matrix obtained from $\Pi^*$ by the factorization \eqref{eq: D factorization}.
Then $D^*$ minimizes the LQG cost \eqref{eq: cost function LQG} among all the aggregation
matrices $D$ of Figure \ref{fig: DP LQG control - aggregation}. This cost is equal to
$J_T^c + \frac{1}{T+1} \sum_{t=0}^T \Tr(X_t^*)$, with $J_T^c$ defined in \eqref{eq: total LQG cost}.
\end{prop}

\subsection{Stationary Problem} \label{sec:StationaryControl problems}

As in Section \ref{sec:Stationary problems} for the filtering problem, we can consider 
the steady-state LQG problem by letting $T \to \infty$ and assuming the 
model \eqref{eq:Global}-\eqref{eq:Global measurements} and the weight matrices $Q$ and $R$ 
in the cost \eqref{eq: cost function LQG} to be time-invariant.
We assume the model to be detectable and stabilizable and the pair
$(A,Q^{1/2})$ to be detectable, in order to be able to implement a stabilizing
LQG controller.
We can take the optimal gains $K^c$ and $K^f$ of the controller and 
the Kalman filter respectively to be also independent of time. 
Following Theorem \ref{thm:FinalTheoremFiltering - steady state} 
and Proposition \ref{prop: LQG result}, we then immediately have the 
following result for the design of the optimal $D$ matrix.

\begin{prop}	\label{prop: LQG result asymptotic}
Let $P$ be the positive semidefinite solution of the following algebraic Riccati equation
\begin{align*}
P = & A^T P A + Q - A^T P B(R+B^T P B)^{-1} B^T PA. 
\end{align*}
Let $L$ be any matrix obtained from the factorization
\[
L^T L = A^T P A + Q  - P.
\]
Let $\Pi^* \succeq 0$, $X^* \succeq 0$, $\Omega^* \succ 0$ be an optimal solution 
for \eqref{eq: SDP cost function-Invariant-1}-\eqref{eq: SDP cost function-Invariant-5},
for this choice of matrix $L$. Suppose that we have $L (\Omega^*)^{-1} C^T \neq 0$.
Let $D^*$ be a matrix obtained from $\Pi^*$ by the factorization \eqref{eq: D factorization}.
Then $D^*$ minimizes the steady-state LQG cost $J_\infty$ 
among all possible matrices $D$ introduced as in Figure \ref{fig: DP LQG control - aggregation},
the corresponding value of the cost is $J_\infty = \Tr(PW) + \Tr(X^*)$.
\end{prop}

\subsection{Numerical Simulations} 
\label{sec:ControlNumericalexperiments} 

We illustrate the above results numerically for $n = 10$ independent scalar systems, 
with states $x_{i,t}$ evolving as first order systems with time-invariant
dynamics \eqref{eq:IndividualGeneral}, where
\begin{align*}
A_1 &= 1.1, A_2 = 0.85, A_3 = 0.84, A_4 = 0.7, A_5 = 0.75, \\
A_6 &= 0.9, A_7 = 0.8, A_8 = 1.05, A_9 = 0.99, A_{10} = 1,
\end{align*}
$C_i = 1$, $W_i = 0.02$ and $V_i = 0.1$ for all $1 \leq i \leq 10$,
and $B$ is a $10 \times 3$ matrix with $B_{ij} = 0$ except for
\begin{align*}
B_{3,1} &= B_{6,1} = B_{9,1} = 1 \\
B_{1,2} &= B_{4,2} = B_{7,2} = B_{10,2} = 1 \\
B_{2,3} &= B_{5,3} = B_{8,3} = 1.
\end{align*}
In other words, the published control signal is
$3$-dimensional, with control input $u_1$ simultaneously actuating
systems $3, 6, 9$, $u_2$ actuating systems $1, 4, 7, 10$ and
$u_3$ actuating systems $2, 5$ and $8$.
We wish to regulate the sum of the states $\sum_{i=1}^{10} x_i$ to $0$, 
hence we take 
$Q$ to be the $10 \times 10$ all-ones matrix 
and $R = I_3$ in \eqref{eq: cost function LQG}. 
We set the privacy parameters to $\delta=0.05$ and $\epsilon=\ln(3)$,
and $\rho_{i}=1$ for $1 \leq i \leq 10$.

To design the differentially private LQG controller with signal aggregation
for the stationary problem, we compute the matrix $L$ of 
Proposition \ref{prop: LQG result asymptotic} and solve the optimization 
problem \eqref{eq: SDP cost function-Invariant-1}-\eqref{eq: SDP cost function-Invariant-5}.
Following the methodology discussed at the end of Section \ref{sec:Examples},
we find that one can take the matrix $D$ to be a $4 \times 10$ matrix
at the matrix factorization stage \eqref{eq: D factorization}.
The corresponding steady-state cost $J_\infty$ is found to be $1.37$,
whereas it is $2.17$ for the input perturbation mechanism (i.e., with
$D = I_{10}$). Hence, signal aggregation results in a significant
improvement.
Figure \ref{fig:quadraticCosts} shows a comparison of the
cost $J_\infty$ for this problem, with the two architectures, 
for different values of $\epsilon$. 
Finally, Figure \ref{fig:LQG sample paths} illustrates the sample
paths obtained under closed-loop control with the differentially
private controllers. 
We see in particular on Figure \ref{fig:2 stage vs input sample path}
that the two-stage architecture provides a much better transient
behavior for the regulated average trajectory (or sum of trajectories)
compared to the input perturbation architecture, in addition to a better
steady-state performance.

\begin{figure}
\centering
\includegraphics[width=\linewidth]{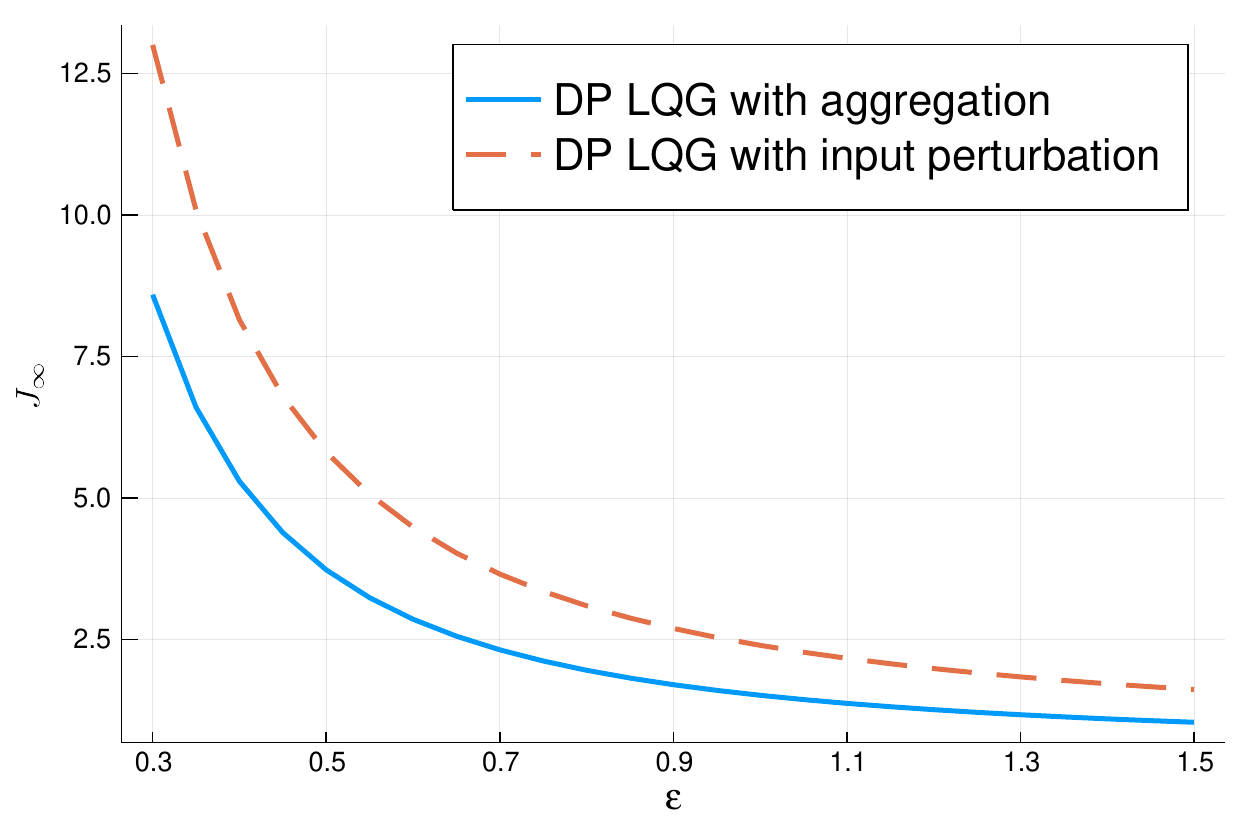}
\caption{Optimal steady-state quadratic cost $J_\infty$ 
as a function of the privacy parameter $\epsilon$,
for the architecture with signal aggregation and for
the input perturbation mechanism. Here $\delta = 0.05$.}
\label{fig:quadraticCosts}
\end{figure}

\begin{figure}  
    \centering
    \begin{subfigure}[b]{\columnwidth}
        \includegraphics[width=\textwidth]{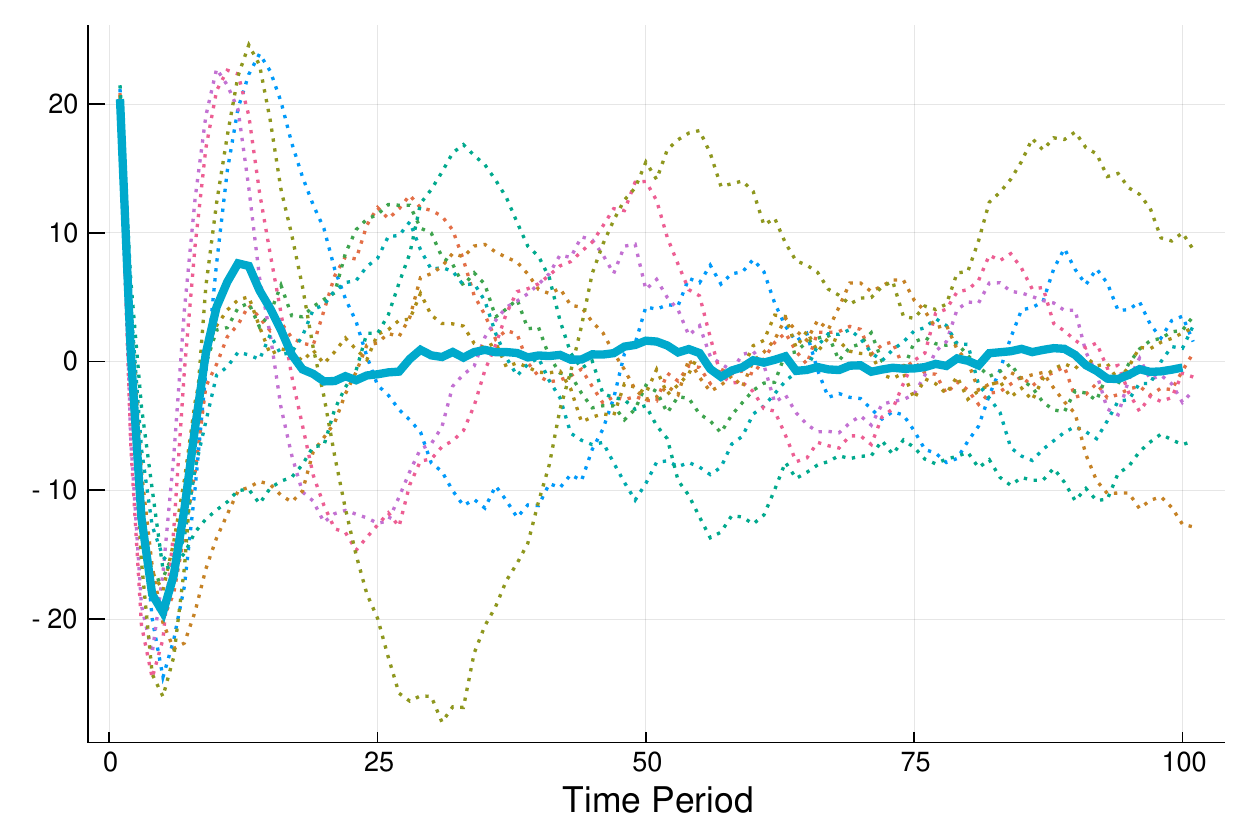}
        \caption{Sample paths for the individual trajectories and the
regulated average trajectory ($\frac{1}{10} \sum_{i=1}^{10} x_{i,t}$, thick solid line) 
using the LQG controller with signal aggregation.}
        \label{fig:sample traj. 2 stage}
    \end{subfigure}
    ~ 
    \begin{subfigure}[b]{\columnwidth}
        \includegraphics[width=\textwidth]{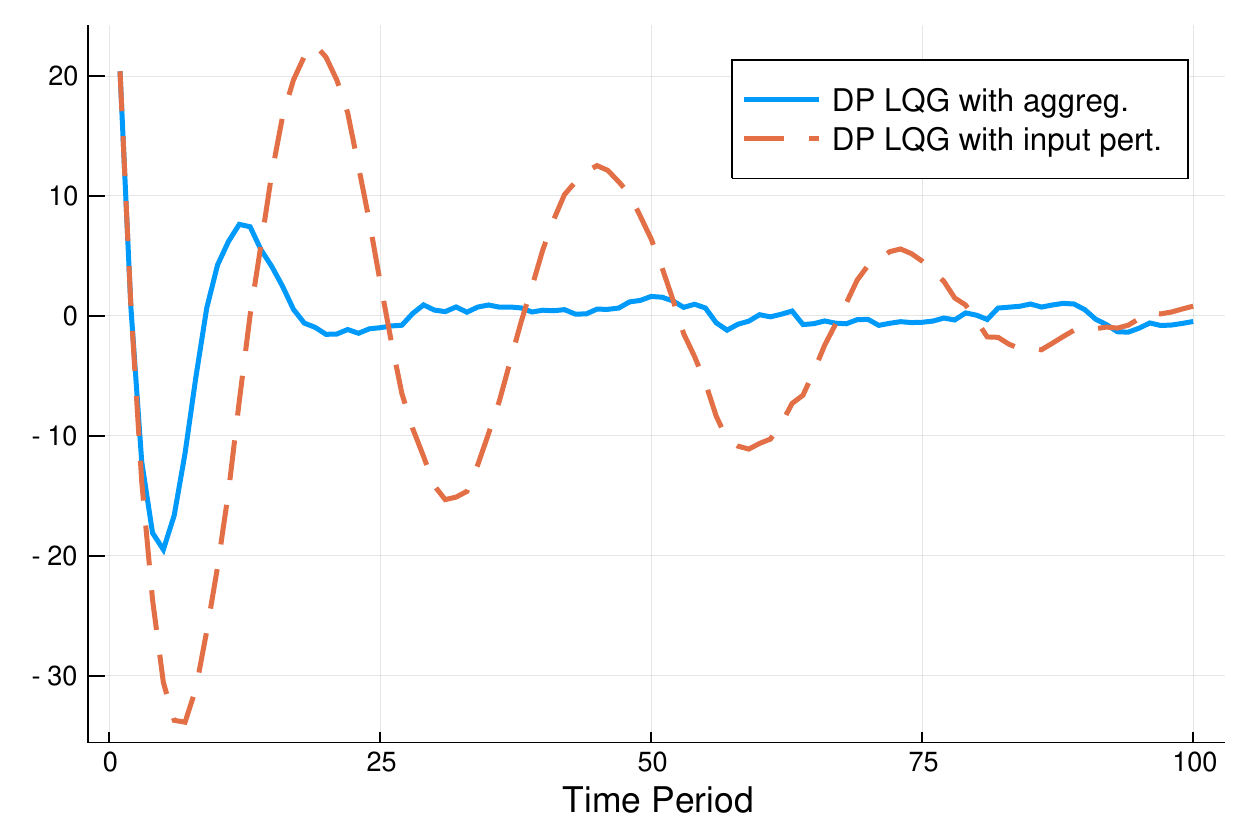}
        \caption{Sample paths of the regulated average trajectories for
the LQG controllers with input perturbation and signal aggregation.}
        \label{fig:2 stage vs input sample path}
    \end{subfigure}
    \caption{Individual and average trajectories under differentially private
    closed-loop control. Here $\epsilon = \ln(3), \delta = 0.05$.
    All individual system states start in the neighborhood of $20$, which is
    also the value used to initialize the Kalman filter estimates for all states.
    }	\label{fig:LQG sample paths}
\end{figure}

\section{Conclusion} \label{sec:Conclusion}

This paper considers the Kalman filtering and LQG optimal control problems
under a differential privacy constraint. We propose  an architecture combining 
an input stage aggregating the individual signals appropriately, the Gaussian 
mechanism to enforce differential privacy and a Kalman filter to reconstruct
the desired estimate. Optimizing the parameters of this architecture can be recast
as an SDP. Examples illustrate the performance improvements 
compared to the input perturbation mechanism, which adds noise directly on the individual
signals. The methodology is then adapted to propose a similar two-stage architecture
for an LQG control problem, where the goal is to compute a shared control broadcasted
to the agent population.
Future research could consider the extension of these ideas to nonlinear systems, 
improving on the input and output perturbation mechanisms
of \cite{LeNy:IJRNC18:dpContraction}.
In addition, since the size of the SDP increases rapidly with the number of agents
(and the time horizon in the non-stationary case), it would be useful to develop 
numerical methods and a problem-specific solver that take advantage of the sparsity 
of the matrices involved in the constraints, as in \cite{Benson2000} for example.

\bibliographystyle{IEEEtran}

\bibliography{diffprivFilter}

\end{document}